\documentclass[reprint,prx,aps,twocolumns,floatfix,
amsmath,notitlepage]{revtex4-2}

\usepackage{amssymb}
\usepackage{graphicx}
\usepackage{graphics}
\usepackage{amsmath}
\usepackage{amsthm}
\usepackage{color}
\usepackage{dsfont}
\usepackage{mathrsfs}

\usepackage{mathtools}
\usepackage[unicode=true,pdfusetitle,bookmarks=true,bookmarksnumbered=false,bookmarksopen=false,breaklinks=false,pdfborder={0 0 0},pdfborderstyle={},backref=false,colorlinks=true]
{hyperref}

\def\pure{{\rm PURE}}
\def\pos{{\rm Pos}}

\def\cptp{{\rm CPTP}}

\def\tho{{\rm TO}}
\def\gp{{\rm GPO}}
\def\gpc{{\rm GPC}}

\def\cto{{\rm CTO}}
\def\distill{{\rm Distill}}
\def\cost{{\rm Cost}}

\def\type{{\rm Type}}

\def\1{\mathds{1}}

\def\>{\rangle}
\def\<{\langle}

\def\id{\mathsf{id}}

\def\mE{\mathcal{E}}

\def\mN{\mathcal{N}}

\def\mP{\mathcal{P}}

\def\mV{\mathcal{V}}

\renewcommand{\qedsymbol}{\nobreak \ifvmode \relax \else
	\ifdim \lastskip<1.5em \hskip-\lastskip \hskip1.5em plus0em
	minus0.5em \fi \nobreak \vrule height0.75em width0.5em
	depth0.25em\fi}

\renewcommand{\geq}{\geqslant}
\renewcommand{\leq}{\leqslant}

\newtheorem{theorem}{Theorem}[section]
\newtheorem{corollary}{Corollary}[section]
\newtheorem{lemma}{Lemma}[section]

\newtheorem{definition}{Definition}[section]
\newtheorem*{remark}{Remark}

\newcommand{\bea}{\begin{eqnarray}}
\newcommand{\eea}{\end{eqnarray}}
\newcommand{\be}{\begin{equation}}
\newcommand{\ee}{\end{equation}}
\newcommand{\ba}{\begin{equation}\begin{aligned}}
\newcommand{\ea}{\end{aligned}\end{equation}}

\newcommand{\bee}{\begin{enumerate}}
\newcommand{\eee}{\end{enumerate}}

\newcommand{\ben}{\begin{enumerate}}
\newcommand{\een}{\end{enumerate}}

\def\be{\begin{equation}}
\def\ee{\end{equation}}

\newcommand{\mU}{\mathcal{U}}

\newcommand{\cov}{{\rm COV}}

\newcommand{\mb}{\mathfrak{B}}
\newcommand{\md}{\mathfrak{D}}
\newcommand{\mf}{\mathfrak{F}}
\newcommand{\ms}{\mathfrak{S}}

\newcommand{\mM}{\mathcal{M}}

\newcommand{\lr}{\rangle\langle}
\newcommand{\la}{\langle}
\newcommand{\ra}{\rangle}
\newcommand{\tr}{{\rm Tr}}

\newcommand{\eps}{\varepsilon}

\newcommand{\mbb}[1]{\mathbb{#1}}

\newcommand{\eqdef}{\coloneqq}

\def\r{\mathbf{r}}
\def\s{\mathbf{s}}
\def\p{\mathbf{p}}
\def\q{\mathbf{q}}
\def\g{\mathbf{g}}

\def\x{\mathbf{x}}

\def\t{\mathbf{t}}
\def\u{\mathbf{u}}

\def\0{\mathbf{0}}

\def\tA{\tilde{A}}
\def\tB{\tilde{B}}

\def\trho{\tilde{\rho}}
\def\tsigma{\tilde{\sigma}}

\newcommand{\GG}[1]{\rm \textcolor{red}{ #1}}
\newcommand{\Gg}[1]{\textcolor{red}{ #1}}

\usepackage[most]{tcolorbox}
\newtcolorbox{myt}[2][]{%
  attach boxed title to top center
               = {yshift=-4pt},
  colback      = blue!5!white,
  colframe     = blue!75!black,
  halign       = flush left,
  fonttitle    = \bfseries\sffamily,
  colbacktitle = blue!65!black,
  title        = #2,#1,
  enhanced,
}
\newtcolorbox{myd}[2][]{%
  attach boxed title to top center
               = {yshift=-4pt},
  colback      = violet!5!white,
  colframe     = violet!75!black,
  halign       = flush left,
  fonttitle    = \bfseries\sffamily,
  colbacktitle = violet!65!black,
  title        = #2,#1,
  enhanced,
}
\newtcolorbox{mye}[2][]{%
  attach boxed title to top center
               = {yshift=-4pt},
  colback      = purple!5!white,
  colframe     = purple!75!black,
  halign       = flush left,
  fonttitle    = \bfseries\sffamily,
  colbacktitle = purple!65!black,
  title        = #2,#1,
  enhanced,
}

\newtcolorbox{myg}[2][]{%
  attach boxed title to top center
               = {yshift=-4pt},
  colback      = green!5!white,
  colframe     = green!75!black,
  halign       = flush left,
  fonttitle    = \bfseries\sffamily,
  colbacktitle = green!65!black,
  title        = #2,#1,
  enhanced,
}

\def\bpm{\begin{pmatrix}}
\def\epm{\end{pmatrix}}

\begin{document}
	
	
	\title{On the Role of Quantum Coherence in  Thermodynamics}

\author{Gilad Gour}\email{gour@ucalgary.ca}
\affiliation{
Department of Mathematics and Statistics, Institute for Quantum Science and Technology,
University of Calgary, AB, Canada T2N 1N4}

	\date{\today}
	
	\begin{abstract}
	We find necessary and sufficient conditions to determine the inter-convertibility of quantum systems under  time-translation covariant evolution, and use it to solve several problems in quantum thermodynamics both in the single-shot and asymptotic regimes. It is well known that  the resource theory of quantum athermality is not reversible, but in PRL 111, 250404 (2013) it was claimed that the theory becomes reversible ``provided a sublinear amount of coherent superposition over energy levels is available". Here we show that if a sublinear amount of coherence among energy levels were considered free, then the resource theory of athermality would become trivial. Instead, we show that by considering a sublinear amount of energy to be free, the theory of athermality becomes reversible for the pure-state case. A proof of the same claim for the mixed-state case is still lacking.  
\end{abstract}

	\maketitle

	\section{Introduction}
Thermodynamics is one of the most prevailing theories in physics with vast applications spreading from its early days focus on steam engines to modern applications in biochemistry, nanotechnology and black hole physics, just to name a few~\cite{GMM2004,BCG+2018,DC2019}. Despite the success of this field, the foundations of thermodynamics remain controversial even today. Not only is there persistent confusion over the relation between the macroscopic and microscopic laws, in particular their reversibility and time-symmetry, there is not even consensus on how best to formulate the second law. Indeed, as the Nobel laureate Percy Bridgman remarked in 1941 “there are almost as many formulations of the Second Law as there have been discussions of it” and the situation hasn’t improved much since then. In recent years, researchers have begun to adopt a new perspective on these foundational problems by reformulating thermodynamics as a resource theory~\cite{CG2019,CFS2016,HO2013}. In 
this approach to thermodynamics, a system that is not in equilibrium with its environment is considered as a resource called \emph{athermality}. Athermality is the fuel that is consumed, in work extraction, computational erasure operation, and other thermodynamical tasks~\cite{JW2000,OHHH2002,RAR+2011,BHO+2013,HO2013,SSP2014,BHN+2015,LKJ+2015,EDR2015,FOR2015,GMN+2015,LJR2015,K2016,GHR+2016,AOP2016,GPS+2016,HFOW2016,WKFR2016,YR2016,LJR2017,MO2017,GJB+2018,BBM2018,H2018,MSP2019,Bernardo2020}.

The resource theoretic approach to thermodynamics focuses on how to quantify a state’s deviation from equilibrium, how to use this for useful tasks in quantum thermodynamics, and what the necessary and sufficient conditions are for one state to be converted to another. In this approach one can consider various notions of state conversion: exact and approximate, single-copy and multiple-copy, with and without the help of a catalyst. 
Such quantum-information techniques lead to many novel insights, particularly given the historical significance of the notion of information for such foundational topics as Maxwell’s demon~\cite{MNV2009}, the thermodynamic reversibility of computation~\cite{B1973,B1982}, Landauer’s principle about the work cost of erasure~\cite{L1961,JW2000}, and Jaynes’s use of maximum entropy principles in deriving statistical mechanics~\cite{J1957a,J1957b}.
Moreover, the resource theoretic approach to thermodynamics demonstrates that the standard formulation of the second law of thermodynamics, as the non-decrease of entropy, is inadequate as a criterion for deciding whether or not a given state conversion is possible. Nonetheless, one can identify a set of measures of the degree of nonequilibrium (including the entropy), such that the state conversion is possible if and only if all of these measures are not increasing~\cite{BHN+2015,GJB+2018}. 

The role of quantum coherence in the resource theory of athermality has several subtleties that were overlooked in some of these works, including the seminal paper~\cite{BHO+2013} that introduced the resource theory of athermality~\footnote{We should point out that the work of~\cite{JWZ+2000} already introduced the resource theory of athermality 13 years earlier}. Specifically, one of the main results of~\cite{BHO+2013} asserts that the free energy ``quantifies the rate at which resource states can be reversibly interconverted asymptotically, provided that a sublinear amount of coherent superposition over energy levels is available, a situation analogous to the sublinear amount of classical
communication required for entanglement dilution". However, it is relatively simple to show (see~\cite{GMS2009} as well as~\eqref{108} below) that the quantum coherence of a pure quantum state $|\psi\ra^{\otimes n}$ grows at most logarithmically with $n$, so that if a sublinear amount of coherence among energy levels were considered free, then the resource theory of athermality would become trivial. What is meant in~\cite{BHO+2013} is that coherence is provided among energy levels that grows sublinearly with $n$. In other words, they assume that the total energy (not coherence) grows sublinearly with $n$, but the proof given in~\cite{BHO+2013} contains some gaps; see Appendix for more specific details.

In this paper we refine this assumption, by considering ``asymptotically negligible resources" to be sequences of quantum states $\{\omega_n\}_{n\in\mbb{N}}$, with $O(\log(n))$ amount of coherence, but whose total energy grows sublinearly with $n$. 
Since the energy of $n$ copies of any (non-zero-energy) state grows linearly with $n$, this assumption is reasonable as it allows for coherence only among energy eigenvectors with asymptotically negligible energy. 
Under this mild assumption we are able to recover the reversibility of the resource theory of quantum athermality in the pure-state regime. 

The paper is organized as follows. After introducing notations and several pertaining preliminary results in Sec.~\ref{preliminaries}, we develop the resource theory of time-translation asymmetry in section~\ref{sectts}, in which we find a simple necessary and sufficient conditions for exact manipulation of quantum coherence. We then apply this result in Sec.~\ref{sec4} for interconversions among athermality states in the single-shot regime. In Sec.~\ref{sec5} we develop the resource theory of quantum athermality in the asymptotic regime, and prove that it is reversible if we allow for sublinear amount of quantum athermality. Finally, in Sec.~\ref{sec6} we end with a discussion and outlook.

\section{Notations and Preliminaries}\label{preliminaries}

In this section we introduce our notations and several results from earlier works. We also present some new results and observations. We denote both quantum systems as well as their corresponding Hilbert spaces by the letters $A$, $B$ ,$A'$, $B'$, and $R$. We will only consider finite dimensional systems and use vertical lines such as $|A|$, $|B|$, to denote the dimension of systems $A$, $B$, respectively. Replicas of a physical system will be denoted with the tilde symbol above them. For example, $\tA$ and $\tB$ are replicas of $A$ and $B$, respectively, and in particular $|\tA|=|A|$ and $|B|=|\tB|$. The set of positive semidefinite matrices acting on system $A$ will be denoted by $\pos(A)$, and quantum states (also called density matrices) in $\pos(A)$ will be denoted by $\md(A)$. The set of pure states in $\md(A)$ will be denoted by $\pure(A)$. The set of all completely positive trace preserving (CPTP) maps, i.e.\ quantum channels, from system $A$ to $B$ are denoted by $\cptp(A\to B)$.

We will use superscripts to indicate actions on subsystems of a composite physical system. For example, let $\rho\in\pos(A)$, $\sigma\in\pos(AB)$, and $\mE\in\cptp(B\to B')$. Then, the notation $\rho^A\sigma^{AB}$ is a short version corresponding to $\left(\rho^A\otimes I^B\right)\sigma^{AB}$, and similarly $\mE^{B\to B'}\left(\sigma^{AB}\right)$ is a short notation of $(\id^A\otimes\mE^{B\to B'})(\sigma^{AB})$, where $\id^A$ is the identity channel. With these notations, the Choi matrix of a channel $\mE\in\cptp(A\to B)$ is defined as
\be
J_\mE^{AB}\eqdef\mE^{\tA\to B}\left(\Phi^{A\tA}\right)\;,
\ee
where $\Phi^{A\tA}\eqdef|\Phi^{A\tA}\lr\Phi^{A\tA}|$, and $|\Phi^{A\tA}\ra\eqdef\sum_{x=1}^m|xx\ra^{A\tA}$ (with $m\eqdef|A|$) is the unnormalized maximally entangled state.

In this paper we only consider physical systems whose Hamiltonians are well defined (i.e. no interactions with other systems) . For example, the Hamiltonians of physical systems $A$ and $B$ will be denoted respectively by $H^A$ and $H^B$.
Moreover, the Hamiltonian of system $A$ (and similarly of system $B$, etc) will be expressed as
\be\label{hamil}
H^A=\sum_{x=1}^ma_x\Pi_x^A
\ee
where $\{a_x\}_{x=1}^m$ are distinct energy eigenvalues, and $\{\Pi_x\}_{x=1}^m$ are orthogonal projectors satisfying $\Pi_x^A\Pi_{y}^A=\delta_{xy}\Pi_x^A$ for all $x,y\in[m]\eqdef\{1,...,m\}$.

\subsection{Notations of Types}

Let  $x^n\eqdef(x_1,...,x_n)$ be a sequence with $n$ elements such that $x_i\in[m]$ for all $i=1,...,n$. For any $z\in[m]$ let $N(z|x^n)$ be the number of elements in the sequence $x^n\eqdef(x_1,...,x_n)$ that are equal to $z$. The \emph{type} of the sequence $x^n$ is a probability vector in $\mbb{R}_{+}^m$ given by
\be
\t(x^n)\eqdef\big(t_1(x^n),...,t_m(x^n)\big)^T\;,
\ee
where 
\be
t_z(x^n)\eqdef\frac{1}{n}N(z|x^n)\;\;\forall\;z\in[m].
\ee

The significance of types to our work comes into play when we consider an i.i.d$\sim \p$ source. In this case, the probability of a sequence $x^n$ drawn from the source is given by (see e.g.~\cite{CT2006})
\ba\label{tpxn}
p_{x^n}\eqdef p_{x_1}\cdots p_{x_n}=2^{-n\big(H(\t(x^n))+D\left(\t(x^n)\|\p\right)\big)}\;,
\ea
where $H(\t(x^n))$ is the Shannon entropy of the type of the sequence $x^n$, and $D\left(\t(x^n)\|\p\right)$ is the Kullback-Leibler divergence between $\t(x^n)$ and $\p$.

We denote  by $\type(n,m)$ the set of all types of sequences in $[m]^n$, and point out that its number of elements is bounded by~\cite{CT2006}
\be\label{typeb}
|\type(n,m)|\leq(n+1)^m\;.
\ee
The set of all sequences $x^n$ of a given type $\t=(t_1,...,t_m)$ will be denoted as $x^n(\t)$. We emphasize that $x^n(\t)$ denotes a \emph{set} of sequences whose type is $\t$, whereas $\t(x^n)$ denotes a \emph{single} probability vector (i.e. the type of a specific sequence $x^n$). The number of sequences
in the set $x^n(\t)$ is given by the combinatorial formula of arranging  $nt_1,...,nt_m$ objects in a sequence; i.e.\
\be\label{xnts}
\left|x^n(\t)\right|={n\choose nt_1,...,nt_m}\eqdef\frac{n!}{\prod_{x=1}^m(nt_x)!}\;.
\ee
The above formula is somewhat cumbersome, but by using Stirling's approximation it can be bounded by~\cite{CT2006}
\be\label{738}
\frac{1}{(n+1)^m}2^{nH(\t)}\leq \left|x^n(\t)\right|\leq 2^{nH(\t)}\;.
\ee

\subsection{Time-Translation Symmetry}

In this subsection we state a few facts about time-translation symmetry. We say that a quantum state $\rho^A$ is time-translation invariant, or quasi-classical, if for all $t\in\mbb{R}$ we have
\be
e^{-iH^At}\rho^A e^{iH^At}=\rho^A\;.
\ee

\begin{definition}
Let $\mE\in\cptp(A\to B)$.
We say that $\mE^{A\to B}$ is time-translation covariant if for all $t\in\mbb{R}$ and all $\rho\in\md(A)$	
\be\label{ttcc}
	\mE^{A\to B}\left(e^{-iH^At}\rho^A e^{iH^At}\right)=e^{-iH^Bt}\mE^{A\to B}\left(\rho^A\right) e^{iH^Bt}\;.
	\ee
The set of all the channels in $\cptp(A\to B)$ that are time-translation covariant is denoted by $\cov(A\to B)$.
\end{definition}

The set of channels $\cov(A\to A)$ contains a special quantum channel known as the \emph{twirling} channel. Expressing the Hamiltonian of system $A$ as in~\eqref{hamil}, the twirling channel on system $A$ is defined by
\be\label{gtwirl}
\mP^{A\to A}\left(\rho^A\right)\eqdef\sum_{x=1}^m\Pi_x^A\rho^A\Pi_x^A\;.
\ee
This twirling channel, also known as the ``pinching channel" (see e.g.~\cite{T2015}), has the property that a state $\rho\in\md(A)$ is quasi-classical iff $\mP(\rho)=\rho$, and if a quantum channel $\mE\in\cov(A\to A)$ then $\mP\circ\mE=\mE\circ\mP$. Moreover, if the Hamiltonian $H^A$ is non-degenerate then
\be
\mP^{A\to A}=\Delta^{A\to A}\;,
\ee
where $\Delta^{A\to A}$ is the completely dephasing channel defined as
\be
\Delta^{A\to A}\left(\rho^A\right)=\sum_{x=1}^{m}\la x|\rho^A|x\ra\; |x\lr x|^A\quad\quad\forall\;\rho\in\md(A)\;.
\ee

The twirling channel can also be used to quantify time-translation asymmetry. For example, the relative entropy distance of a quantum state $\rho\in\md(A)$ to its twirled state $\mP(\rho)$ is a time-translation asymmetry (sometimes referred to as coherence) measure given by
\be\label{coherence}
C(\rho)\eqdef D\left(\rho\big\|\mP(\rho)\right)=H\big(\mP(\rho)\big)-H(\rho)\;,
\ee
where $D(\rho\|\sigma)\eqdef\tr[\rho\log\rho]-\tr[\rho\log\sigma]$ is the Umegaki relative entropy and $H(\rho)\eqdef-\tr[\rho\log\rho]$ is the von-Neumann entropy. The above function is non-increasing under time-translation covariant operations, and achieves its maximal value of $\log|A|$ on the maximally coherent state $|+\ra\eqdef\frac1{\sqrt{|A|}}\sum_{x=1}^{|A|}|x\ra$.

For $n$ copies of system $A$, we will denote by $\mP_n\in\cov(A^n\to A^n)$ the pinching channel associated with the total Hamiltonian $H^{A^n}$ given by
\ba
H^{A^n}&\eqdef H^A\otimes I^A\otimes\cdots\otimes I^A+I^A\otimes H^A\otimes\cdots\otimes I^A\\
&+\cdots+I^A\otimes\cdots\otimes I^A\otimes H^A\;.
\ea
With these notations we have $\mP=\mP_1$. 
In~\cite{GMS2009} it was shown that $C(\rho^{\otimes n})$ grows logarithmicly with $n$ (see also~\eqref{108} below) and in particular,
\be\label{016}
\lim_{n\to\infty}\frac1nC\left(\rho^{\otimes n}\right)=0\;.
\ee

\subsection{The resource theory of athermality}

In this subsection we review the resource theory of athermality. We put emphasis on some subtleties that are quite often overlooked in the existing literature. In particular, we distinguish between thermal operations and closed thermal operations.
Moreover, we prove some new results. Specifically, to the author's knowledge, all the lemmas and theorems presented here are new.

\subsubsection{Free states and athermality resource-states}

The free states in the resource theory of athermality correspond to physical systems that are in thermal equilibrium with their surrounding. For a heat bath that is held at a fixed inverse temperature $\beta\eqdef\frac1{k_BT}$, the free thermodynamics state, 
\be
\gamma^B\eqdef\frac{e^{-\beta H^B}}{\tr\left[e^{-\beta H^B}\right]}\;,
\ee
is the thermal equilibrium state known as the Gibbs state (here $H^B$ is the Hamiltonian associated with the heat bath). We will always use the greek letter $\gamma$ to indicate a Gibbs state. For example, the notation $\gamma^A$, $\gamma^{A'}$, and $\gamma^B$, correspond to the Gibbs states of systems $A$, $A'$, and $B$, respectively. Moreover, the joint Gibbs state of two non-interacting systems $A$ and $B$ will be denoted simply by $\gamma^{AB}=\gamma^A\otimes\gamma^B$.

In the QRT of athermality, every physical system that can be used as a resource has a well defined Hamiltonian. Therefore, a physical system $A$ cannot be characterized just by a density matrix $\rho\in\md(A)$ since the resourcefulness of the state depends also on the Hamiltonian of the system, $H^A$. For this reason, every thermodynamic state in quantum thermodynamics comprise of a quantum state $\rho\in\md(A)$ acting on the Hilbert space $A$, and a time-independent Hamiltonian $H^A\in\pos(A)$ that governs the dynamics of the quantum system $A$. That is, a \emph{state} of athermality can be characterized by a pair $(\rho^A,H^A)$. This is indeed the characterization used extensively in literature. 

From the resource-theoretic perspective, this characterization of an athermality state has several drawbacks. First, it is not invariant under an energy shift of the form $H^A\mapsto H^A+cI^A$, where $c\in\mbb{R}$ is some constant. Indeed, the choice of setting the minimal energy of a system to be zero is somewhat arbitrary. Second, the resourcefulness of the state $\rho^A$ is determined in relation to its deviation from the Gibbs state $\gamma^A$ of system $A$. 
Therefore, it seems more natural to characterize  athermality states (i.e. the ``objects" of this theory) by pairs of the form $(\rho^A,\gamma^A)$. Note that all the relevant information about the Hamiltonian $H^A$ is contained in the Gibbs state $\gamma^A$ which is invariant under energy shifts.

\subsubsection{Free Operations}

The set of free operations relative to a background heat bath at temperature $T$ comprise of three basic steps:
\begin{enumerate}
\item Thermal equilibrium. Any subsystem $B$, with Hamiltonian $H^{B}\in\pos(B)$, can be prepared in its thermal Gibbs state $\gamma^B$.
\item Conservation of energy.  Unitary operation on a composite physical system that commutes with the total Hamiltonian can be implemented. 
\item Discarding subsystems. It is possible to trace over any subsystem (with a well defined Hamiltonian) of a composite system.
\end{enumerate}

Any CPTP map comprising of the above three steps is called a \emph{thermal operation}. Any thermal operation $\mE\in\cptp(A\to A)$ can be expressed as
\be\label{reato}
\mE^{A\to A}(\rho^A)=\tr_{B}\left[\mU^{AB\to AB}\left(\rho^A\otimes\gamma^B\right)\right]
\ee
where $\mU\in\cptp(AB\to AB)$ is a unitary channel that is Gibbs preserving; i.e. 
\be
\mU^{AB\to AB}\left(\gamma^{AB}\right)=\gamma^{AB}\;,
\ee 
where $\gamma^{AB}=\gamma^A\otimes\gamma^B$. In the lemma below we show that $\cptp(A\to A')$ with $|A|\neq |A'|$ also contains thermal operations.

\begin{myt}{}
\begin{lemma}
Let $AB$, $A'B'$ be two composite physical systems with $|AB|=|A'B'|$, and let $\mU\in\cptp(AB\to A'B')$ be a Gibbs preserving unitary channel; that is,
$$
\mU^{AB\to A'B'}\left(\gamma^{AB}\right)=\gamma^{A'B'}\;.
$$
Then, the map (defined on all $\omega\in\md(A)$)
\be\label{reato}
\mN^{A\to A'}\left(\omega^{A}\right)\eqdef\tr_{B'}\left[\mU^{AB\to A'B'}\left(\omega^A\otimes\gamma^B\right)\right]
\ee
is a thermal operation. 
\end{lemma}
\end{myt}
\begin{proof}
Consider the joint Gibbs state $\gamma^{ABA'B'}\eqdef \gamma^{AB}\otimes \gamma^{A'B'}$ and let $\mV\in\cptp(ABA'B'\to ABA'B')$ be the unitary channel given by
\be
\mV^{ABA'B'\to ABA'B'}\eqdef \mU^{AB\to A'B'}\otimes \mU^{*A'B'\to AB}\;.
\ee 
Observe that $\mV^{ABA'B'\to ABA'B'}$ preserves the joint Gibbs state $\gamma^{ABA'B'}$.
Hence, the channel
\begin{align*}
&\tr_{ABB'}\left[\mV^{ABA'B'\to ABA'B'}\left(\omega^A\otimes\gamma^{BA'B'}\right)\right]\\
&=\tr_{ABB'}\left[\mU^{AB\to A'B'}\left(\omega^A\otimes\gamma^B\right)\otimes\mU^{*A'B'\to AB}\big(\gamma^{A'B'}\big)\right]\\
&=\tr_{B'}\left[\mU^{AB\to A'B'}\left(\omega^A\otimes\gamma^B\right)\right]\\
&=\mN^{A\to A'}\left(\omega^{A}\right)\;,
\end{align*} 
is a thermal operation. This completes the proof.
\end{proof}

We denote by $\tho(A\to A')$ the set of all thermal operations in $\cptp(A\to A')$.
For fixed systems $A$ and $A'$ the set $\tho(A\to A')$ is in general not closed and not convex. It stems from the fact that the dimensions of systems $B$ and $B'$ as appear in~\eqref{reato} are unbounded. Therefore, it will be convenient to define the closure of $\tho(A\to A')$, denoted by $\cto(A\to A')$, as a set of channels in $\cptp(A\to A')$ with the property that $\mE\in\cto(A\to A')$ if and only if there exists a sequence of thermal operations $\{\mE_k\}_{k\in\mbb{N}}$, where each $\mE_k\in\tho(A\to A')$ and 
\be
\lim_{k\to\infty}\mE_k=\mE\;.
\ee
By definition, the set $\cto(A\to A')$ is closed.  We now prove that it is also convex.

\begin{myt}{}
\begin{theorem}\label{ctoconvex}
The set $\cto(A\to A')$ is convex.
\end{theorem}
\end{myt}
\begin{proof}
We start by showing that $\tho(A\to A')$ is closed under convex combination with rational coefficients. Specifically, let 
\be
\mN^{A\to A'}\eqdef\sum_{x=1}^\ell \frac{m_x}m\mN_x^{A\to A'}
\ee
where each $m_x\in\mbb{N}$, $m\eqdef\sum_{x=1}^\ell m_x$, and each $\mN_x\in\tho(A\to A')$. Since each $\mN_x$ is a thermal operation it can be expressed as
\be
\mN_x^{A\to A'}\left(\omega^A\right)\eqdef\tr_{B_x'}\left[\mU_x^{AB_x\to A'B_x'}\left(\omega^{A}\otimes\gamma^{B_x}\right)\right]\;,
\ee
for some systems $B_x,B_x'$ and some unitary channel $\mU_x\in\cptp(AB_x\to A'B_x')$. For each $y\in[m]$,
let $k_y$ be the integer in $[\ell]$ 
satisfying
\be
\sum_{x=1}^{k_y-1}m_x\leq y<\sum_{x=1}^{k_y}m_{x}\;,
\ee
and define 
\be
B\eqdef\bigoplus_{y=1}^mB_{k_y}\;,\;\;B'\eqdef\bigoplus_{y=1}^mB_{k_y}'\quad\text{and}\quad\gamma^B\eqdef\frac1m\bigoplus_{y=1}^m\gamma^{B_{k_y}}\;.
\ee
Finally, for any $\eta^{AB}=\bigoplus_{y=1}^m\eta^{AB_{k_y}}\in\md(AB)$ we define the action of the unitary channel $\mU\in\cptp(AB\to A'B')$ as
\be
\mU^{AB\to A'B'}(\eta^{AB})=\bigoplus_{y=1}^m\mU_{k_y}^{AB_{k_y}\to A'B_{k_y}'}(\eta^{AB_{k_y}})\;.
\ee
With these definitions we get
\ba
&\tr_{B'}\left[\mU^{AB\to A'B'}\left(\omega^{A}\otimes\gamma^{B}\right)\right]\\
&=\frac1m\sum_{y=1}^m\tr_{B'_{k_y}}\left[\mU^{AB_{k_y}\to A'B'_{k_y}}\left(\omega^{A}\otimes\gamma^{B_{k_y}}\right)\right]\\
&=\frac1m\sum_{x=1}^\ell m_x\tr_{B_x'}\left[\mU_x^{AB_x\to A'B_x'}\left(\omega^{A}\otimes\gamma^{B_x}\right)\right]
\ea
where in the last line we used the fact that for any $x\in[\ell]$ there exist $m_x$ values of $y\in[m]$ for which $k_y=x$. Finally, observe that the RHS of the equation above is precisely $\mN^{A\to A'}(\omega^A)$. Therefore, $\mN^{A\to A'}$ is a thermal operation.
This completes the proof that any rational convex combination of thermal operations is a thermal operation. 

To prove the convexity of $\cto(A\to A')$ let $\{\mM_x\}_{x=1}^k$ be $k$ channels in $\cto(A\to A')$ and let 
\be
\mM^{A\to A'}\eqdef\sum_{x=1}^k p_x\mM_x^{A\to A'}
\ee
be a convex combination of the $k$ channels $\{\mM_x\}$. For each $n\in\mbb{N}$ let $\mM_x^{(n)}\in\tho(A\to A')$ be such that 
$\lim_{n\to\infty}\mM_x^{(n)}=\mM_x$, and let $\{p_x^{(n)}\}_{x=1}^k$ be a rational probability distribution with the property that $\lim_{n\to\infty}p_{x}^{(n)}=p_x$. Now, from the previous argument we have that for all $n\in\mbb{N}$ the rational convex combination
\be
\sum_{x=1}^k p_x^{(n)}\mM_x^{(n)}
\ee
is in $\tho(A\to A')$. Therefore, by definition, the limit
\be
\lim_{n\to\infty}\sum_{x=1}^k p_x^{(n)}\mM_x^{(n)}=\mM
\ee
is in $\cto(A\to A')$. This completes the proof.
\end{proof}

Every thermal operation $\mE\in\cptp(A\to A')$ has two key properties: 
\ben
\item  $\mE^{A\to A'}$ is \emph{Gibbs preserving operation} ($\gp$); that is, $\mE(\gamma^A)=\gamma^{A'}$.
\item $\mE^{A\to A'}$ is time-translation covariant; i.e. $\mE\in\cov(A\to A')$.
\een
The set of all Gibbs preserving operations in $\cptp(A\to A')$ will be denoted by $\gp(A\to A')$, and those that are Gibbs preserving covariant (GPC) quantum channels (i.e.\ channels that satisfy the above two properties) will be denoted by $\gpc(A\to A')$. In what follows, we will also use the notations
\be
(\rho^A,\gamma^A)\xrightarrow{\;\;\mf\;\;}(\sigma^B,\gamma^B)\;,
\ee
to indicate that $(\rho^A,\gamma^A)$ can be converted to $(\sigma^B,\gamma^B)$ by the free operations $\mf$. In our context, $\mf$ can stand for thermal operations, closed thermal operations (CTO), GPC, and GPO. Since GPC form a closed set of operations we have for any two systems $A$ and $A'$
\ba
\tho(A\to A')\subset\cto(A\to A')&\subset\gpc(A\to A')\\
&\subset\gp(A\to A')\;.
\ea

We now show that the pinching channel is a thermal operation.
\begin{lemma}\label{lem2}
Consider the pinching channel $\mP\in\cptp(A\to A)$ associated with the Hamiltonian of system $A$. Then, $\mP\in\tho(A\to A)$.
\end{lemma}

\begin{proof}
Expressing the Hamiltonian of system $A$ as in~\eqref{hamil}, the pinching channel $\mP\in\cptp(A\to A)$ can be written as a mixture of unitaries of the form (see for example~\cite{T2015}) 
\be\label{27}
\mP(\rho)=\frac1m\sum_{x=1}^mU_x\rho U_x^*\quad\quad\forall\;\rho\in\md(A)\;,
\ee
where
\be
U_x^A\eqdef\sum_{x'\in[m]}e^{\frac{2\pi ixx'}{m}}\Pi_{x'}^A\;.
\ee
Clearly, each of the $m$ unitaries $\{U_x^A\}$ commutes with the Hamiltonian $H^A$. Therefore, each unitary channel $\mU_x\in\cptp(A\to A)$, defined via $\mU_x(\cdot)\eqdef U_x(\cdot)U_x^*$, is a thermal operation. In the proof of Theorem~\ref{ctoconvex} we showed that any rational convex combination of thermal operations is itself a thermal operation. Therefore, the mixture of unitaries in~\eqref{27} is a thermal operation. This completes the proof.
\end{proof}

\subsubsection{Quasi-Classical Athermality}

We say that an athermality state $(\rho^A,\gamma^A)$ is quasi-classical if $\rho^A$ and $\gamma^A$ commute; that is, $\rho$ is diagonal in the energy eigenbasis of system $A$. In this case, we will denote the athermality state  $(\rho^A,\gamma^A)$ as $(\p^A,\g^A)$, where $\p^A$ and $\g^A$ are probability vectors consisting of the diagonals of $\rho^A$ and $\gamma^A$, respectively. In this quasi-classical regime, for two athermality states $(\p^A,\g^A)$ and $(\p^B,\g^B)$ we have (see Theorem~5 in~\cite{JWZ+2000})
\be\label{19}
(\p^A,\g^A)\xrightarrow{\cto}(\q^B,\g^B)\;\;\iff\;\;(\p^A,\g^A)\succ(\q^B,\g^B)\;,
\ee
where $\succ$ denotes \emph{relative majorization}. Relative majorization is a pre-order defined between two pairs of probability vectors. Specifically, we say that $(\p^A,\g^A)$ relatively majorizes $(\q^B,\g^B)$ (and write it as in the equation above) if there exists a column stochastic matrix $E$ such that $\q^B=E\p^A$ and $\g^B=E\g^A$. Relative majorization has several characterizations including a geometrical one given by Lorenz curves and testing regions (see e.g.~\cite{R2016}). 

If the Hamiltonian of system $A$ is fully degenerate (we will say in this case that the Hamiltonian is trivial) then $H^A=cI^A$ for some constant $c\geq 0$ and the corresponding Gibbs state, 
\be
\g^A=\u^{(m)}\eqdef\frac1m\bpm 1\\\vdots\\ 1\epm\;,
\ee 
is the $m$-dimensional uniform probability vector. 
We say that two athermality states,
$(\p^A,\g^A)$ and $(\p^B,\g^B)$, are equivivalent, and write
\be
 (\p^A,\g^A)\sim(\p^B,\g^B)
\ee
 if both $(\p^A,\g^A)\succ(\p^B,\g^B)$ and $(\p^B,\g^B)\succ(\p^A,\g^A)$. One of the remarkable properties of quasi-classical thermodynamics is that a dense set of athermality states are equivalent to states with a trivial (i.e.\ zero) Hamiltonian~\cite{GT2021}. Specifically, let $\g=(g_1,...,g_m)^T$ be the Gibbs state of system $A$ and suppose that its components $\{g_x\}$ are rational. Then, there exists $k_1,...,k_m\in\mbb{N}$ such that for each $x\in[m]$ we have $g_x=\frac{k_x}{k}$, where $k\eqdef\sum_{x=1}^mk_x$ is the common denominator. With such a Gibbs state, for any probability vector $\p=(p_1,...,p_m)^T$ we have that~\cite{GT2021}
 \be
 (\p,\g)\sim(\r,\u^{(k)})\quad\text{where}\quad\r\eqdef\bigoplus_{x=1}^mp_x\u^{(k_x)}\;.
 \ee
 The above equivalence indicates that athermality of the the quasi-classical system $A$ can be fully characterized by the \emph{non-uniformity} of the vector $\r$, with $\r=\frac1k(1,...,1)^T$ being the least resourceful and $\r=(1,0,...,0)^T$ being the most resourceful. Therefore, in the quasi-classical regime the resource theory of athermality is essentially equivalent to the resource theory of non-uniformity, also known as the resource theory of informational non-equilibrium~\cite{GMN+2015}.

\subsubsection{The Golden Unit of Athermality}

A ``golden unit" of a resource theory is a constituent of a resource that can be used to measure the resource very much like ebits are used to measure entanglement.
Due to the equivalence between athermality and non-uniformity in the quasi-classical regime, we can use units of non-uniformity to measure the athermality of a given state. 
Specifically, we can take the golden unit to have the form
$(|0\lr 0|^A,\u^A)$. This golden unit is equivalent to~\cite{WW2019}
\be
(|0\lr 0|^A,\u^A)\sim\left(|0\lr 0|^X,\u^X_m\right)\;,
\ee
where $X$ is a two-dimensional classical system, $m\eqdef|A|$, and
\be\label{334}
\u_m^X\eqdef \frac1m|0\lr 0|^X+\frac{m-1}{m}|1\lr 1|^X\;.
\ee
Therefore, we can always consider the golden unit to be a qubit. Moreover, note that $\u_m^X$ is well defined even if $m$ is not an integer. 
This can help simplifying certain expressions, and we will therefore consider also the states $\left(|0\lr 0|^X,\u^X_m\right)$ with $m\in\mbb{R}_+$.

\subsubsection{Cost and Distillation}

We will denote by $\mf$ the free operations of the resource theory of athermality. We will consider three cases in which $\mf=\cto$, $\mf=\gp$, and $\mf=\gpc$. In either of these cases, we define the conversion distance as
\be\label{32}
d_\mf\big((\rho^A,\gamma^A)\to(\sigma^B,\gamma^B)\big)\eqdef \min_{\mE\in\mf(A\to B)}\frac12\left\|\sigma^B-\mE\left(\rho^A\right)\right\|_1\;.
\ee
The conversion distance measures the closest distance (in trace norm) that $\rho^A$ can reach $\sigma^B$  by using only free operations. 
For any $\eps>0$ and $\rho,\gamma\in\md(A)$, this conversion distance can be used to define the $\eps$-single-shot distillable athermality as
\begin{align}\label{33}
&\distill_{\mf}^{\eps}\left(\rho,\gamma\right)
\eqdef\log\sup_{0<m\in\mbb{R}}\nonumber\\
&\Big\{m\;:\;d_\mf\Big(\left(\rho^A,\gamma^A\right)\to\left(|0\lr 0|^X,\u_m^X\right)\Big)\leq\eps\Big\}\;.
\end{align}

The asymptotic distillation of an athermality state $(\rho,\gamma)$ is defined as
\ba
&\distill_{\mf}\left(\rho,\gamma\right)
\eqdef\lim_{\eps\to 0^+}\sup_{\ell,n\in\mbb{N}}\\
&\left\{\frac \ell n\;:\;d_\mf\Big(\left(\rho^{\otimes n},\gamma^{\otimes n}\right)\to\left(|0\lr 0|^{\otimes \ell},\u_2^{\otimes \ell}\right)\Big)\leq\eps\right\}\;,
\ea
where $\u_2$ is the 2-dimensional maximally mixed state. The single-shot and asymptotic distillation rates are related by
\be
\distill_\mf(\rho,\gamma)=\lim_{\eps\to 0^+}\limsup_{n\to\infty}\frac1n\distill^\eps_\mf\left(\rho^{\otimes n},\gamma^{\otimes n}\right)\;.
\ee
We point out that $\distill_\mf(\rho,\gamma)$ has the property that for any $k\in\mbb{N}$
\ba\label{39}
&\frac1k\distill_\mf\left(\rho^{\otimes k},\gamma^{\otimes k}\right)\\
&=\lim_{\eps\to 0^+}\limsup_{n\to\infty}\frac1{nk}\distill^\eps_\mf\left(\rho^{\otimes nk},\gamma^{\otimes nk}\right)\\
&\leq \lim_{\eps\to 0^+}\limsup_{n'\to\infty}\frac1{n'}\distill^\eps_\mf\left(\rho^{\otimes n'},\gamma^{\otimes n'}\right)\\
&=\distill_\mf(\rho,\gamma)\;.
\ea

Similarly, the conversion distance can be used to define the $\eps$-single-shot athermality cost as
\begin{align}
&\cost_{\mf}^{\eps}\left(\rho,\gamma\right)
\eqdef\log\inf_{0<m\in\mbb{R}}\nonumber\\
&\Big\{m\;:\;d_\mf\Big(\left(|0\lr 0|^X,\u_m^X\right)\to \left(\rho^A,\gamma^A\right)\Big)\leq\eps\Big\}\;.
\end{align}
The asymptotic athermality cost of the state $(\rho,\gamma)$ is defined as
\ba
&\cost_{\mf}\left(\rho,\gamma\right)
\eqdef\lim_{\eps\to 0^+}\inf_{m,n\in\mbb{N}}\\
&\left\{\frac mn\;:\;d_\mf\Big(\left(|0\lr 0|^{\otimes m},\u_2^{\otimes m}\right)\to\left(\rho^{\otimes n},\gamma^{\otimes n}\right)\Big)\leq\eps\right\}\;.
\ea
The single-shot and asymptotic athermality costs are related by
\be
\cost_\mf(\rho,\gamma)=\lim_{\eps\to 0^+}\liminf_{n\to\infty}\frac1n\cost^\eps_\mf\left(\rho^{\otimes n},\gamma^{\otimes n}\right)\;.
\ee
For the case that $\mf=\gp$ all the quantities above have
relatively simple closed formulas. In the single-shot regime we have ~\cite{WW2019}
\ba\label{34}
&\distill_{\gp}^{\eps}\left(\rho,\gamma\right)=D_{\min}^\eps(\rho\|\gamma)\\
&\cost_\gp^\eps(\rho,\gamma)=D_{\max}^\eps(\rho\|\gamma)\;,
\ea
where $D_{\min}^\eps$ is the Hypothesis testing divergence defined as
\ba
D_{\min}^\eps(\rho\|\gamma)\eqdef\min_{0\leq\Lambda\leq I^A}\Big\{\tr[\gamma\Lambda]\;:\;\tr[\Lambda\rho]\geq 1-\eps\Big\}
\ea
and $D_{\max}^\eps$ is the smoothed max relative entropy defined as
\be
D_{\max}^\eps(\rho\|\gamma)\eqdef\min\Big\{D_{\max}(\rho'\|\gamma)\;:\;\frac12\|\rho-\rho'\|_1\leq\eps\Big\}
\ee
and $D_{\max}(\rho\|\gamma)\eqdef\log\min\{t\geq 0\;:\;t\gamma\geq\rho\}$. In the asymptotic regime, under GPO, the resource theory of athermality is reversible as reflected by the equality
\be
\distill_{\mf}\left(\rho,\gamma\right)=\cost_{\mf}\left(\rho,\gamma\right)=D(\rho\|\gamma)\;,
\ee
where $D(\rho\|\gamma)\eqdef\tr[\rho\log\rho]-\tr[\rho\log\gamma]$ is the Umegaki relative entropy.

The hypothesis testing divergence that appear above in the formula for the single-shot distillable athermality is neither additive nor subadditive under tensor products. Instead it satisfies a weaker type of subadditivity given in the lemma below.
\begin{lemma}
Let $\eps>0$, $\rho,\gamma\in\md(A)$, and $\rho',\gamma'\in\md(A')$. Then,
\be\label{47}
D_{\min}^\eps\left(\rho\otimes\rho'\big\|\gamma\otimes\gamma'\right)\leq D_{\min}^\eps(\rho\|\gamma)+D_{\max}\left(\rho'\|\gamma'\right)
\ee
\end{lemma}

\begin{proof}
By definition, 
\be
2^{-D_{\min}^\eps\left(\rho\otimes\rho'\|\gamma\otimes\gamma'\right)}=\min_{\tr\left[\left(\rho\otimes\rho'\right)\Lambda\right]\geq 1-\eps}\tr\left[\left(\gamma\otimes\gamma'\right)\Lambda\right]
\ee
where the minimum is over all effects $\Lambda\in\pos(AA')$ that satisfies $\Lambda\leq I^{AA'}$. The key idea is to use the inequality
\be
\gamma'\geq 2^{-D_{\max}(\rho'\|\gamma')}\rho'\;.
\ee
This inequality follows directly from the definition of $D_{\max}(\rho'\|\gamma')$. Therefore, from the above two equations we get
\begin{align}
&2^{-D_{\min}^\eps\left(\rho\otimes\rho'\|\gamma\otimes\gamma'\right)}\nonumber\\
&\quad\quad\quad\geq2^{-D_{\max}(\rho'\|\gamma')}\min_{\tr\left[\left(\rho\otimes\rho'\right)\Lambda\right]\geq 1-\eps}\tr\left[\left(\gamma\otimes\rho'\right)\Lambda\right]\nonumber\\
&\quad\quad\quad=2^{-D_{\max}(\rho'\|\gamma')}\min_{\tr\left[\rho\Gamma\right]\geq 1-\eps}\tr\left[\gamma\Gamma\right]\;,\label{49}
\end{align}
were the second minimum is over all effects $\Gamma\in\pos(A)$ of the form 
\be\label{050}
\Gamma\eqdef\tr_{A'}\left[\left(I^A\otimes\rho'\right)\Lambda^{AA'}\right]\;.
\ee 
By removing the constraint~\eqref{050} on $\Gamma$ and taking instead the minimum over all operators $0\leq\Gamma\leq I^{A}$ we get that the minimization $\min_{\tr\left[\rho\Gamma\right]\geq 1-\eps}\tr\left[\gamma\Gamma\right]$ equals by definition to $2^{-D_{\min}(\rho\|\gamma)}$. Therefore, since the removal of the constraint ~\eqref{050} can only decrease the second minimization in~\eqref{49} we conclude that
\be
2^{-D_{\min}^\eps\left(\rho\otimes\rho'\|\gamma\otimes\gamma'\right)}\geq 2^{-D_{\max}^\eps\left(\rho'\|\gamma'\right)}2^{-D_{\min}^\eps\left(\rho\|\gamma\right)}\;.
\ee
This completes the proof.
\end{proof}

\subsubsection{The Relative Entropy of Athermality}

The relative entropy of athermality of a state $(\rho^A,\gamma^A)$ is defined in terms of the Umegaki relative entropy  as
$
D(\rho^A\|\gamma^A)
$. This function quantifies the athermality of the state $(\rho^A,\gamma^A)$, and is related to the free energy via  $D(\rho\|\gamma)=\beta\big(F(\rho^A)-F(\gamma^A)\big)$, where $F(\rho^A)$ is the free energy of $\rho^A$ (see~\cite{BHO+2013}). 
The relative entropy distance can also be expressed as
\ba\label{18p26}
D\left(\rho\|\gamma\right)&=-H(\rho)-\tr\left[\rho\log\gamma\right]\\
&=-H(\rho)-\tr\left[\mP(\rho)\log\gamma\right]\\
&=D\left(\mP(\rho)\big\|\gamma\right)+H\big(\mP(\rho)\big)-H(\rho)\\
&=D\left(\mP(\rho)\big\|\gamma\right)+C(\rho)\;.
\ea
That is, the athermality of the state $(\rho,\gamma)$ can be decomposed into two components:
\begin{enumerate}
\item Its nonuniformity that is quantified by $D\left(\mP(\rho)\big\|\gamma\right)$.
\item Its asymmetry (or coherence between energy eigenspaces) that is quantified by the time-translation asymmetry measure $C(\rho)$.
\end{enumerate}
Moreover, since the regularization of the coherence vanishes (see~\eqref{016}) we conclude that
\be\label{18p27}
\lim_{n\to\infty}\frac1nD\left(\mP_n\left(\rho^{\otimes n}\right)\big\|\gamma^{\otimes n}\right)=D\left(\rho\|\gamma\right)\;.
\ee

\section{Time-Translation Symmetry}\label{sectts}

We start by developing the resource theory of time-translation asymmetry. Specifically, we will provide necessary and sufficient conditions for state conversions in this model. As we mentioned in the preliminary section, we are only considering in this paper physical systems whose Hamiltonians are well defined. It turns out that the degeneracy of these Hamiltonians play an important role in the manipulation of asymmetry.

\subsection{Degenerate vs Non-degenerate Hamiltonians}\label{sec:exactasy}

Let $H^{A}$ and $H^{B}$ be the Hamiltonians of two systems $A$ and $B$, of dimensions  $m\eqdef |A|$ and $n\eqdef |B|$. The Hamiltonians can be expressed in their spectral decomposition as
	\be\label{hamilab}
	H^A=\sum_{x=1}^ma_x|x\lr x|^A\quad\text{and}\quad H^B=\sum_{y=1}^nb_x|y\lr y|^B\;,
	\ee
	where $\{a_x\}$ and $\{b_y\}$ are the energy eigenvalues of $H^A$ and $H^B$, respectively.

\begin{definition}
We say that the Hamiltonians $H^A$ and $H^B$, as defined in~\eqref{hamilab}, are \emph{relatively non-degenerate} if for all $x,x'\in[m]$ and $y,y'\in[n]$ we have
\be
a_x-a_{x'}=b_y-b_{y'}\quad\Rightarrow\quad x=x'\;\text{ and }\;y=y'\;.
\ee
If the condition above does not hold we say that the Hamiltonians are relatively degenerate.
\end{definition}

Note that if $H^A$ and $H^B$ are relatively non-degenerate, then each of them is also non-degenerate. For example, suppose $H^A$ is degenerate with $a_x=a_{x'}$ for some $x\neq x'\in[m]$. Then, for $y=y'$ we get $a_x-a_{x'}=0=b_y-b_{y'}$ even though $x\neq x'$. Therefore, relative non-degeneracy is a stronger notion than non-degeneracy. In fact, relative non-degeneracy of $H^A$ and $H^B$ is equivalent to the non-degeneracy of the joint Hamiltonian $H^{AB}=H^A\otimes I^B+I^A\otimes H^B$.
Moreover, in the generic case in which $H^A$ and $H^B$ are arbitrary (chosen at random) the Hamiltonians are relatively non-degenerate. 	
For this case, time-translation covariant channels have a very simple characterization.

\begin{myt}{}
\begin{theorem}\label{thm1}
{\rm Let $A$ and $B$ be two physical systems with relatively non-degenerate Hamiltonians. Then, $\mN\in\cptp(A\to B)$ is a time-translation covariant channel if and only if
\be
\mN^{A\to B}=\Delta^{B\to B}\circ\mN^{A\to B}\circ\Delta^{A\to A}\;,
\ee
where $\Delta^{A\to A}$ and $\Delta^{B\to B}$ are the completely dephasing channels of systems $A$ and $B$, respectively. In other words, for physical systems with relatively non-degenerate Hamiltonians only classical channels are time-translation covariant.}
\end{theorem}
\end{myt}
\begin{proof}
We start by expressing~\eqref{ttcc} in the Choi representation. 	Specifically, by replacing $\rho$ in~\eqref{ttcc} with the unnormalized maximally entangled state $\Phi^{A\tA}$, the RHS becomes
\be\label{58}
e^{-iH^Bt}\mE^{\tA\to B}\left(\Phi^{A\tA}\right)e^{iH^Bt}=e^{-iH^Bt}J^{AB}_\mE e^{iH^Bt}\;,
\ee
and the LHS of~\eqref{ttcc} can be expressed as
\ba\label{59}
\mE^{\tA\to B}\left(e^{-iH^{\tA} t}\Phi^{A\tA}  e^{iH^{\tA} t}\right)&=\mE^{\tA\to B}\left(e^{-iH^A t}\Phi^{A\tA} e^{iH^{A} t}\right)\\
&=e^{-iH^At}J^{AB}_{\mE}e^{i{H}^At}\;,
\ea
where in the first equality we used the fact that $|\Phi^{A\tA}\ra=\frac{1}{\sqrt{|A|}}\sum_{x=1}^{|A|}|x\ra^A|x\ra^{\tA}$ (here $\{|x\ra^A\}$ and $\{|x\ra^{\tA}\}$ are eigenbases of $H_A$ and $H_{\tA}$, repectively) has the property that 
\be
e^{-iH^{\tA} t}\big|\Phi^{A\tA}\big\ra=e^{-iH^{A} t}\big|\Phi^{A\tA}\big\ra\;.
\ee
Hence, by equating~\eqref{58} with~\eqref{59} we get that in the Choi representation the condition on $\mE$ in~\eqref{ttcc} is equivalent to	
\be
	e^{-i{H}^At}\otimes e^{iH^Bt}J^{AB}_{\mE}e^{i{H}^At}\otimes e^{-iH^Bt}=J^{AB}_{\mE}\;.
	\ee	
	Now, substituting into the above equation $J^{AB}_\mE=\sum_{x,x',y,y'}c_{xyx'y'}|x\lr x'|\otimes|y\lr y'|$ (where $c_{xyx'y'}$ are some coefficients) gives
	\be
	c_{xyx'y'}e^{i\left(-a_{x}+a_{x'}+b_{y}-b_{y'}\right)t}=c_{xyx'y'}\quad\forall\;t\in\mbb{R}\;.
	\ee
	Hence, $c_{xyx'y'}=0$ unless 
	\be
	a_x-b_y=a_{x'}-b_{y'}\;.
	\ee
	Combining this with the fact that $H^A$ and $H^B$ are relatively non-degenerate we get that 	$c_{xyx'y'}=0$ unless $x=x'$ and $y=y'$. Hence, the Choi matrix $J_\mE^{AB}$ can be expressed as
	\be
	J^{AB}_\mE=\sum_{x,y}c_{xyxy}|x\lr x|^A\otimes|y\lr y|^B
	\ee
	so that $\mE^{A\to B}$ is a classical channel. This completes the proof. 
\end{proof}	
	
	We consider now the interesting case in which $A=B$. In this case, we have in particular $H^A=H^B$ so we cannot apply the result above to this case.

\begin{definition}	
	Let $H^A$ be the Hamiltonian of a system $A$ with energy eigenvalues $\{a_x\}_{x=1}^m$. We say that $H^A$ has a non-degenerate Bohr spectrum if it has the property  that for any $x,y,x',y'\in[m]$ 
\ba
a_x-a_y=a_{x'}-a_{y'}\quad\iff\quad &x=x'\text{ and }y=y'\\
\text{or}\quad &x=y\text{ and }x'=y'\;;\nonumber
\ea
that is, there are no degeneracies in the nonzero
differences of the energy levels of $H^A$.
\end{definition}

Observe that almost all Hamiltonians have a non-degenerate Bohr spectrum (i.e. the set of all Hamiltonians that do not have a non-degenerate Bohr spectrum is of measure zero). Therefore, the results below that involves Hamiltonian with non-degenerate Bohr spectrum will apply to almost all systems. 
Such time-translation covariant channels with respect to  non-degenerate Bohr spectrums have the following characterization.

\begin{lemma}\label{lem1}
Let $H^A$ be a Hamiltonian with a non-degenerate Bohr spectrum, $\mE\in\cptp(A\to A)$, and $m\eqdef|A|$. Then, $\mE\in\cov(A\to A)$ iff there exists a conditional probability distribution $\{p_{y|x}\}_{x,y\in[m]}$, and an $m\times m$ positive semidefinite matrix $Q$ (with components denoted as $\{q_{xy}\}_{x,y\in[m]}$) whose diagonal components are $q_{xx}=p_{x|x}$ (for all $x\in[m]$) such that the Choi matrix of $\mE^{A\to A}$ is given by 
\ba\label{bohr}
	J^{A\tA}_\mE&=\sum_{x,y\in[m]}p_{y|x}|x\lr x|^A\otimes |y\lr y|^{\tA}\\
	&+\sum_{\substack{x\neq y\\x,y\in[m]}}q_{xy}|x\lr y|^A\otimes |x\lr y|^{\tA}\;.
	\ea
\end{lemma}

\begin{proof}
Following the same lines as in Theorem~\ref{thm1}, by replacing $H^B$ with $H^A$ everywhere, we get that a quantum channel $\mE\in\cptp(A\to A)$ is time-translation covariant iff its Choi matrix $J^{A\tA}_\mE=\sum_{x,x',y,y'}c_{xyx'y'}|x\lr x'|\otimes|y\lr y'|$ satisfies
$c_{xyx'y'}=0$ unless 
	\be
	a_x-a_y=a_{x'}-a_{y'}\;.
	\ee
Since $H^A$ is generic (i.e. has non-degenerate Bohr spectrum) we get that the Choi matrix $J^{A\tA}_\mE$ corresponds to a time-translation covariant channel iff $c_{xyx'y'}=0$ unless $x=x'$ and $y=y'$, or $x=y$ and $x'=y'$. Denoting by $p_{y|x}\eqdef c_{xyxy}$ and $q_{xy}\eqdef c_{xxyy}$ we conclude that $\mE\in\cptp(A\to A)$ is time-translation covariant channel iff $J_\mE^{A\tA'}$ has the form~\eqref{bohr}.  Since $J_\mE^{A\tA}\geq 0$ we have in particular that each $p_{y|x}\geq 0$, and the condition that the marginal $J^A=I^A$ implies that $\sum_{y}p_{y|x}=1$ for all $x=1,...,m$. Note that the two terms on the RHS of~\eqref{bohr} have orthogonal support. Therefore, $J^{AB}\geq 0$ iff both $p_{y|x}\geq 0$ for all $x$ and $y$, and $Q\geq 0$. This completes the proof.
\end{proof}

\begin{remark} 
Observe that even if the spectrum of the Hamiltonian $H^A$ has degeneracies, any quantum channel $\mE\in\cptp(A\to A)$ whose Choi matrix has the form~\eqref{bohr} is necessarily time-translation covariant. Therefore, several of the results below will also be useful for Hamiltonians with degenerate spectrum.
\end{remark}

\subsection{Exact Interconversions}

In this subsection we consider the exact conversion of one state to another under time-translation covariant channels. Specifically, let $\{|x\ra^A\}_{x\in[m]}$ be the energy-eigenbasis of an Hamiltonian $H^A$, and let
\be\label{12}
\rho^A=\sum_{x,x'\in[m]}r_{xx'}|x\lr x'|^A\;\text{ and }\;\sigma^A=\sum_{x,x'\in[m]}s_{xx'}|x\lr x'|^A
\ee
be two density matrices in $\md(A)$ with components $\{r_{xx'}\}$ and $\{s_{xx'}\}$, respectively.

\begin{myt}{}	
\begin{theorem}\label{ttcs}
{\rm Let $\rho,\sigma\in\md(A)$ be as in~\eqref{12} and suppose that $r_{xx'}\neq0$ for all $x,x'\in[m]$, and the Hamiltonian $H^A$ has a non-degenerate Bohr spectrum.
Then, the following statements are equivalent:
\ben
\item There exists $\mE\in\cov(A\to A)$ such that $\sigma=\mE(\rho)$. 
\item The $m\times m$ matrix $Q$, with components
\be\label{1369}
q_{xy}\eqdef\begin{cases} 
\min\left\{1,\frac{s_{xx}}{r_{xx}}\right\} & \text{ if }x=y\\
\frac{s_{xy}}{r_{xy}}& \text{ otherwise.}
\end{cases}
\ee
is positive semidefinite.
\een
Moreover, the second statement implies the first statement even if the Hamiltonian $H^A$ has a degenerate Bohr spectrum.}
\end{theorem}
\end{myt}
\begin{remark}
In the proof below we will see that if $r_{xy}=0$ for some off diagonal terms (i.e. $x\neq y$) then $s_{xy}$ must also be zero. However, in this case, we will see that for any $x\neq y\in[m]$ with $r_{xy}=0$, the components of $q_{xy}$ can be arbitrary. This means that in this case the condition becomes cumbersome, as we will need to require that there \emph{exists} $Q$ as defined above but with no restriction on the components $q_{xy}$ for which $r_{xy}=0$.
\end{remark}
\begin{proof}	
	 From Lemma~\ref{lem1} it follows that there exists $\mE\in\cov(A\to A)$ such that $\sigma=\mE(\rho)$ iff there exists a conditional probability distribution $\{p_{y|x}\}$, and an $m\times m$ positive semidefinite matrix $Q$, such that
	\ba
	\sigma=\mE(\rho)&=\tr_A\left[J^{A\tA}_\mE(\rho^T\otimes I^{\tA})\right]\\
	&=\sum_{x, y}p_{y|x}r_{xx}|y\lr y|+\sum_{x\neq y}q_{xy}r_{xy}|x\lr y|
	\ea
	That is, $\sigma=\mE(\rho)$ iff 
	\ba
	&s_{yy}=\sum_{x=1}^mp_{y|x}r_{xx}\quad\quad\forall\;y\in[m]\quad\text{and}\\
	&s_{xy}=q_{xy}r_{xy}\quad\quad\forall\;x\neq y\in[m]\;.
	\ea
	Hence, for the off diagonal terms, $s_{xy}=0$ whenever $r_{xy}=0$. Since
	we assume that  that all the off-diagonal terms of $\rho$ are non-zero, i.e.\ $r_{xy}\neq 0$ for $x\neq y$, there is no freedom left in the choice of the off diagonal terms of $Q$ and we must have $q_{xy}=\frac{s_{xy}}{r_{xy}}$.
Since $Q$ must be positive semidefinite we will maximize its diagonal terms $\{p_{x|x}\}_{x=1}^m$ given the constraint that $s_{yy}=\sum_{x=1}^mp_{y|x}r_{xx}$. This constraint immediately gives $s_{yy}\geq p_{y|y}r_{yy}$ so that we must have $p_{y|y}\leq\frac{s_{yy}}{r_{yy}}$. Clearly, we also have $p_{y|y}\leq 1$ so we conclude that
	\be
	p_{y|y}\leq\min\left\{1,\frac{s_{yy}}{r_{yy}}\right\}\;.
	\ee
	Remarkably, this condition is sufficient since there exists conditional probabilities $\{p_{y|x}\}$, with both $p_{y|y}=\min\left\{1,\frac{s_{yy}}{r_{yy}}\right\}$ and 	$s_{yy}=\sum_{x=1}^mp_{y|x}r_{xx}$.	Indeed,  for simplicity set $r_x\eqdef r_{xx}$ and $s_x\eqdef s_{xx}$, and define
	\be\label{coedf}
	p_{y|x}\eqdef
	\begin{cases}
	\min\left\{1,\frac{s_{x}}{r_{x}}\right\} &\text{if }x=y\\
	\frac{1}{\mu r_x}(s_{y}-r_{y})_+(r_x-s_x)_+ & \text{otherwise}
	\end{cases}
	\ee
	where
	\be
	\mu\eqdef\sum_{y\in[m]}(s_{y}-r_y)_+=\frac12\|\s-\r\|_1\;,
	\ee
	and we used the notation $(s_{y}-r_y)_+\eqdef s_y-r_y$ if $s_y\geq r_y$ and $(s_{y}-r_y)_+\eqdef 0$ if $s_y< r_y$. Clearly, $p_{y|x}\geq 0$, and it is straightforward to check that $\sum_{y=1}^mp_{y|x}=1$ and $s_{y}=\sum_{x=1}^mp_{y|x}r_{x}$; that is, the above conditional probability distribution satisfies all the required conditions. This completes the proof for the case that $H^A$ has non-degenerate Bohr spectrum.

Finally, if $H^A$ has degenerate Bohr spectrum and $Q\geq 0$ then	we still get that the Choi matrix of the form~\eqref{bohr} (with $p_{y|x}$ as in~\eqref{coedf} and $q_{xy}$ as in~\eqref{1369}) corresponds to a quantum channel $\mE\in\cptp(A\to A)$ with the property that $\sigma=\mE(\rho)$.
As discussed below the proof of Lemma~\ref{lem1}, all channels with a Choi matrix of the form~\eqref{bohr} are time-translation covariant. Hence, $\mE\in\cov(A\to A)$.
	This completes the proof.
\end{proof}	
\begin{remark}
In the proof above we saw that if $r_{xy}=0$ for some $x\neq y$ then $\sigma=\mE(\rho)$ for some $\mE\in\cov(A\to A)$ only if $s_{xy}=0$. This in particular implies that if $\rho$ has a block diagonal form $\rho=\begin{pmatrix}\trho & \0\\
\0 & \0
\end{pmatrix}$, and if it can be converted by a time-translation covariant channel to $\sigma$, then $\sigma$ must have the form $\sigma=\begin{pmatrix}\tsigma & \0\\
\0 & D
\end{pmatrix}$ where $D$ is some diagonal matrix.
\end{remark}

As an example for the theorem above, consider the qubit case in which
both $$\rho=\begin{pmatrix} a & z\\\bar{z} &1-a  &\end{pmatrix}\quad\text{and}\quad\sigma=\begin{pmatrix} b & w\\\bar{w} &1-b  &\end{pmatrix}$$ are qubit states.
W.l.o.g.\ suppose that $a\geq b$ (we can always rearrange the order of the diagonals of $\rho$ and $\sigma$ by a permutation in $\cov(A\to A)$).
In this case the matrix $Q$ can be expressed as
\be
Q=\begin{pmatrix}
\frac ba\; &\; \frac wz\\
\; &\; \\
\frac{\bar{w}}{\bar{z}}\; &\; 1
\end{pmatrix}\;.
\ee
Therefore, $Q\geq 0$ iff
\be\label{coddd}
\frac ba\geq\left|\frac wz\right|^2\;.
\ee
Observe that if $\rho$ is a pure state, so that $|z|=\sqrt{a(1-a)}$, then the above equation holds iff
$|w|^2\leq b(1-a)$. Now, since $\sigma\geq 0$ we have $|w|^2\leq b(1-b)$ so that $1-\frac{|w|^2}{b}\geq b$. Therefore,
for any $a$ in the range
\be
a\in\left[b,1-\frac{|w|^2}{b}\right]
\ee
we get both $|w|^2\leq b(1-a)$ and $a\geq b$. That is, for any mixed state $\sigma$ there exists a pure state $\psi$ that can be converted to $\sigma$. On the other hand, if $\sigma$ is pure (i.e. $|w|^2=b(1-b)$) and $\rho$ arbitrary qubit, then the condition in~\eqref{coddd} becomes
\be
\left|z\right|^2\geq a(1-b)\;.
\ee
Since $\rho\geq 0$ we also have $|z|^2\leq b(1-b)$. Combining both inequalities we find that the only way $\rho$ can be converted to a pure qubit state $\sigma$ is if
$b=a$ (since $a\geq b$ was the initial assumption) and $|z|^2=a(1-a)$. That is, $\rho$ is a pure state itself, and up to a diagonal unitary equals to $\sigma$. Hence, pure coherence cannot be obtained from mixed coherence, and deterministic interconversion among inequivalent pure resources is not possible.

The example above shows that there is no unique ``golden unit" that can be used as the ultimate resource in two dimensional systems. Instead, any pure resource (i.e. pure state that is not an energy eigenstate) is maximal in the sense that there is no other resource that can be converted into it (up to the equivalence class of diagonal unitaries). However, the set of all pure qubit resources is maximal (i.e. any mixed state can be reached from some pure state by time-translation covariant operations). We now show that this latter property holds in general.

\begin{corollary}\label{puretomix}
Let $\sigma\in\md(A)$ be an arbitrary state, and denote by $p_x\eqdef\la x|\sigma|x\ra$ the diagonal elements of $\sigma$ in the energy eigenbasis $\{|x\ra\}_{x=1}^m$ of system $A$. Then, the pure quantum state
\be
|\psi\ra\eqdef\sum_{x=1}^m\sqrt{p_x}|x\ra
\ee
can be converted to $\sigma$ by a time-translation covariant channel.
\end{corollary}

\begin{proof}
Observe that the diagonal elements $Q$ are all 1, and the off diagonal terms are given by
\be
q_{xy}=\frac{s_{xy}}{\sqrt{p_xp_y}}\quad\quad\forall\;x,y\in[m]\;,\;x\neq y.
\ee
Therefore, we can express $Q=D_\p^{-1}\sigma D_{\p}^{-1}$, where $D_{\p}$ is the diagonal matrix whose diagonal is $(\sqrt{p_1},...,\sqrt{p_m})$. Since $D_\p>0$ and $\sigma\geq 0$ it follows that $Q\geq 0$.
\end{proof}

\section{Quantum Athermality in the Single-Shot Regime}
\label{sec4}

In Sec.~\ref{sec:exactasy} we saw that if $A$ and $B$ are two physical systems with relatively non-degenerate Hamiltonians, then a quantum channel $\mN\in\cptp(A\to B)$ is time translation covariant if and only if it is classical. Since thermal operations are time-translation covariant, it follows that for relatively non-degenerate Hamiltonians thermal operations must be classical. This observation has the following consequence.

 \begin{corollary}\label{cor:rgsg}
 Let $\mf$ be either $\cto$ or $\gpc$, and let $A$ and $B$ be two physical systems with relatively non-degenerate Hamiltonians. Let $(\rho^A,\gamma^A)$ and $(\sigma^B,\gamma^B)$ be two athermality states on system $A$ and $B$, and $\r^A$, $\s^B$, $\g^A$, and $\g^B$, be the probability vectors whose components are the elements on the diagonals of $\rho^A$, $\sigma^B$, $\gamma^A$ and $\gamma^B$, respectively. Then, the following are equivalent:
 \ben
 \item $\left(\rho^A,\gamma^A\right)\xrightarrow{\mf}\left(\sigma^B,\gamma^B\right)$.
 \item $\sigma^B$ is diagonal in the energy-eigenbasis and $\left(\r^A,\g^A\right)\succ\left(\s^B,\g^B\right)$.
 \een
 \end{corollary}

 \begin{remark}
Note that in the generic case of relatively non-degenerate Hamiltonians, $\gpc$ can only destroy the coherence between the energy levels of the input state $\rho^A$. In this case, coherence cannot be manipulated, but only destroyed. 
\end{remark}


Consider the conversion of one athermality state $(\rho^A,\gamma^A)$ to another athermality state
$(\sigma^{B},\gamma^{B})$ under any of the free operations, $\mf$, discussed above. Such a conversion is equivalent to a conversion with the same input and output Gibbs states, since appending a Gibbs state is a reversible thermal operation.  To see this explicitly, observe first that
\ba
&\left(\rho^A,\gamma^A\right)\xleftrightarrow{\;\mf\;\;}\left(\rho^A\otimes\gamma^{B},\gamma^{AB}\right)\\&\left(\sigma^{B},\gamma^{B}\right)\xleftrightarrow{\;\mf\;\;}\left(\gamma^A\otimes\sigma^{B},\gamma^{AB}\right)\;,
\ea
where $\mf$ is one of the four sets $\tho$, $\cto$, $\gpc$, and $\gp$, and the symbol $\xleftrightarrow{\;\mf\;\;}$ indicates conversion under $\mf$ in both directions. 
Therefore, the conversion of $\left(\rho^A,\gamma^A\right)$ to $\big(\sigma^{B},\gamma^{B}\big)$ is equivalent to the conversion of $\big(\rho^A\otimes\gamma^{B},\gamma^{AB}\big)$
to the state $\big(\gamma^A\otimes\sigma^{B},\gamma^{AB}\big)$. Note that the latter conversion has the same input and output Gibbs state $\gamma^{AB}$.
Therefore, interconversions among states with the same Gibbs state (and in particular with $|A|=|B|$) is general enough to capture also interconversions with $|B|\neq |A|$ (as long as we do not impose some additional non-degeneracy constraints); see Fig.~\ref{appending}. 

\begin{figure}[h]\centering    \includegraphics[width=0.45\textwidth]{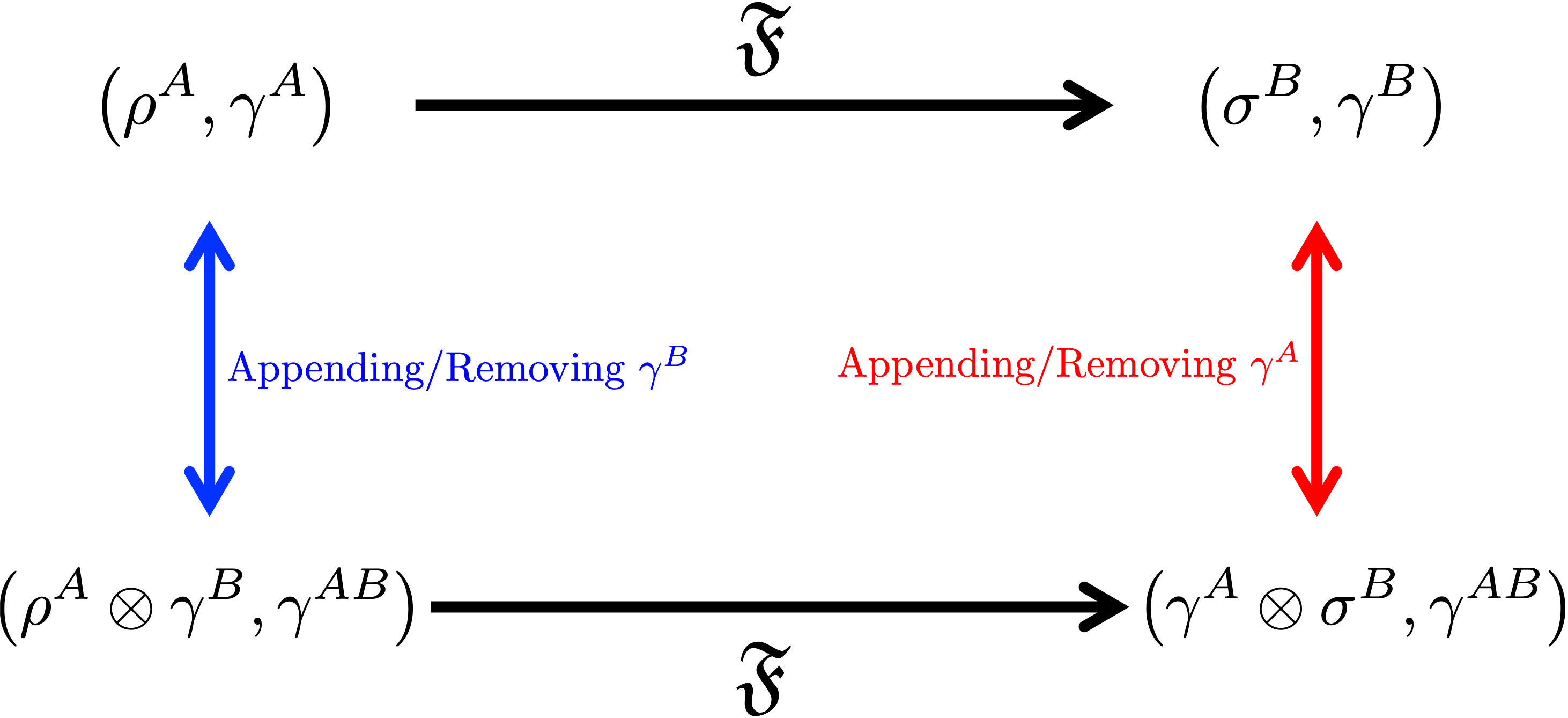}
  \caption{\linespread{1}\selectfont{\small Equivalence of conversions. The top conversion with two different Gibbs states $\gamma^A$ and $\gamma^B$ is equivalent to the bottom conversion with the same Gibbs state $\gamma^{AB}$.}}
  \label{appending}
\end{figure}

We now focus on interconversions among states that are all in $\md(A)$, and unless necessary, will drop the system superscript $A$ from the states. However, we will assume that the Hamiltonian $H^A$ has a non-degenerate Bohr spectrum. This will reduce a bit from the generality of the results, however, as discussed above, this is the generic case and almost all Hamiltonians having such a spectrum.

\subsection{Exact Conversions}

Consider a conversion of the form $(\rho,\gamma)\xrightarrow{\gpc}(\sigma,\gamma)$, where $\rho,\sigma,\gamma\in\md(A)$, and \emph{all} the off-diagonal terms of $\rho$ are non-zero.
In this case, Theorem~\ref{ttcs} states that $\rho$ can be converted to $\sigma$ by a time-translation covariant channel iff the matrix $Q$ as defined in~\eqref{1369} is positive semidefinite.  Since CGP channels are in particular covariant under the time-translation group, the condition $Q\geq 0$ is a necessary (but not sufficient) condition for
$(\rho,\gamma)\xrightarrow{\gpc}(\sigma,\gamma)$. To get the full necessary and sufficient conditions, let $J^{AB}$ be the Choi matrix of a time-translation covariant channel $\mE\in\cov(A\to A)$ that satisfies $\mE(\rho)=\sigma$ and $\mE(\gamma)=\gamma$. Recall that the Choi matrix of such a channel has the form (cf.~\eqref{bohr})
\be
	J^{A\tA}=\sum_{x, y}p_{y|x}|x\lr x|^A\otimes |y\lr y|^{\tA}+\sum_{x\neq y}\frac{s_{xy}}{r_{xy}}|x\lr y|^A\otimes |x\lr y|^{\tA}\;,
	\ee
where $P=(p_{y|x})$ is some column stochastic matrix, and we assumed that the off diagonal terms of $\rho$ are non-zero. Let $\r$ and $\s$ be the probability vectors consisting of the diagonals of $\rho$ and $\sigma$, respectively, and identify the diagonal matrix $\gamma$ with the Gibbs vector $\g$ consisting of its diagonal. Then, the Choi matrix above facilitates such a channel $\mE$ iff it is positive semidefinite \emph{and} 
\be\label{1621}
P\r=\s\quad\text{and}\quad P\g=\g\;.
\ee
Note that the above condition implies that $(\r,\g)\succ(\s,\g)$, however, it is not sufficient since we also require that $J^{A\tA}\geq 0$. This latter condition is equivalent to the requirement that the matrix obtained by replacing the diagonal elements of $Q$ (as defined in~\eqref{1369}) with $\{p_{x|x}\}_{x\in[m]}$ is positive semidefinite. 
We summarize these considerations in the following lemma.

\begin{lemma}\label{lem1841}
Let $(\rho^A,\gamma^A)$ and $(\sigma^A,\gamma^A)$ be two athermality states of a system $A$, whose Hamiltonian $H^A$ has a non-degenerate Bohr spectrum. Using the same notations as in~\eqref{12}, suppose that $r_{xy}\neq 0$ for all $x\neq y$. Then, the following statements are equivalent:
\ben
\item 
$
 (\rho^A,\gamma^A)\xrightarrow{\gpc}(\sigma^A,\gamma^A)
$.
\item There exists a column stochastic matrix $P$ that satisfies both~\eqref{1621} and the matrix
 \be\label{matrix18}
\sum_{x=1}^mp_{x|x}|x\lr x|+\sum_{x\neq y\in[m]}\frac{s_{xy}}{r_{xy}}|x\lr y|\geq 0\;.
 \ee
 \een
 Moreover, the second statement implies the first statement even if the Hamiltonian $H^A$ has a degenerate Bohr spectrum.
\end{lemma}

The Lemma above does not provide much computational simplification over the results in~\cite{GJB+2018},  since determining the existence of such a column stochastic matrix $P$ is itself a semidefinite programming (SDP) problem. However, the significance of this lemma is that it makes the role of quantum coherence in such conversions of athermality much more apparent, as demonstrated by the following theorem.
Moreover, we will see below that in the qubit case the lemma above provides a simple criterion for exact inter-conversions under GPC.

\begin{myt}{}
\begin{theorem}\label{thm31}
{\rm Let $(\rho^A,\gamma^A)$ and $(\sigma^A,\gamma^A)$ be two quantum athermality states of dimension $m\eqdef|A|$, whose Hamiltonian $H^A$ has a non-degenerate Bohr spectrum.
For any $x,y\in[m]$ let $r_{xy}\eqdef\la x|\rho|y\ra$ and $s_{xy}\eqdef\la x|\sigma|y\ra$ be, respectively, the $xy$-component of $\rho$ and $\sigma$ in the energy-eigenbasis. 
Suppose that $r_{xy}\neq 0$ for all $x,y\in[m]$ and that $r_{xx}=s_{xx}$ for all $x\in[m]$. Then, the following statements are equivalent:
\ben
\item $(\rho^A,\gamma^A)\xrightarrow{\gpc}(\sigma^A,\gamma^A)$.
 \item 
$
Q^A\eqdef I^A+\sum_{x\neq y\in[m]}\frac{s_{xy}}{r_{xy}}|x\lr y|^A\geq 0
$.
 \een
 Moreover, the second statement implies the first statement even if the Hamiltonian $H^A$ has a degenerate Bohr spectrum.}
\end{theorem}
\end{myt}

\begin{remark}
The condition in the theorem above that $\rho$ and $\sigma$ have the same diagonals means that $\rho$ and $\sigma$ have the same non-uniformity and they only differ by their coherence (asymmetry) properties. In fact, observe that the condition $Q^A\geq 0$ is identical to the condition given in~Theorem~\ref{ttcs} for the case that the diagonals of $\rho$ and $\sigma$ are the same. Therefore, in this case we have $(\rho^A,\gamma^A)\xrightarrow{\rm CGP}(\sigma^A,\gamma^A)$ if and only if $\rho^A$ can be converted to $\sigma^A$ by time-translation covariant operations. In particular, the Gibbs state, $\gamma^A$, does not play a role in such conversions since $\rho$ and $\sigma$ have the same non-uniformity (i.e. same diagonals).
\end{remark}

\begin{proof}
Since the diagonals of $\rho$ and $\sigma$ are the same, we get that if $Q^A\geq 0$ then by taking the stochastic matrix $P$ to be the identity matrix, all the conditions of Lemma~\ref{lem1841} are satisfied so that $(\rho^A,\gamma^A)\xrightarrow{\rm CGP}(\sigma^A,\gamma^A)$. Conversely, if $(\rho^A,\gamma^A)\xrightarrow{\rm CGP}(\sigma^A,\gamma^A)$ then by Lemma~\ref{lem1841} there exists a stochastic matrix $P$ with a diagonal $\{p_{x|x}\}$ that satisfies~\eqref{matrix18}. By adding the positive semidefinite matrix $\sum_{x=1}^m(1-p_{x|x})|x\lr x|$ to the matrix in~\eqref{matrix18} we get that also $Q^A\geq 0$. This completes the proof.
\end{proof}

The theorem above has the following consequence.

\begin{corollary}\label{puretomixgpc}
Let $\rho\in\md(A)$ be an arbitrary state, and denote by $p_x\eqdef\la x|\rho|x\ra$ the diagonal elements of $\rho$ in the energy eigenbasis $\{|x\ra\}_{x=1}^m$ of system $A$. Then, 
\be
\left(\psi^A,\gamma^A\right)\xrightarrow{\gpc}\left(\rho^A,\gamma^A\right)
\ee
where $|\psi^A\ra\eqdef\sum_{x=1}^m\sqrt{p_x}|x\ra$.
\end{corollary}

\begin{proof}
Since $\psi^A$ and $\rho^A$ have the same diagonals, it follows from Theorem~\ref{thm31} and the discussion above that $\left(\psi^A,\gamma^A\right)\xrightarrow{\gpc}\left(\rho^A,\gamma^A\right)$ if and only if $\psi^A$ can be converted to $\rho^A$ by time-translation covariant operations. The latter conversion is possible due to Corollary~\ref{puretomix}. This completes the proof.
\end{proof}

Lemma~\ref{lem1841} can also be used to give the precise conditions for interconversions under GPC of qubit athermality states. For this purpose,
let $\rho,\sigma,\gamma\in\md(A)$ with $|A|=2$. Denote
\be\label{18p79}
\rho=\begin{pmatrix}
r & a \\
\bar{a} & 1-r
\end{pmatrix}
\;\text{,}\;\;
\sigma=\begin{pmatrix}
s & b \\
\bar{b} & 1-s
\end{pmatrix}
\;,\;\;
\gamma=\begin{pmatrix}
g & 0 \\
0 & 1-g
\end{pmatrix}\;.
\ee 
We also denote the diagonals of the matrices above by $\r\eqdef(r,1-r)^T$, $\s\eqdef(s,1-s)^T$ and $\g=(g,1-g)^T$. We would like to find the conditions under which $(\rho,\gamma)\xrightarrow{\gpc}(\sigma,\gamma)$. Recall that if $a=0$ then we must have $b=0$ since GPC cannot generate coherence between energy levels. Therefore, the case $a=0$ has already been covered by the quasi-classical regime. Note also that the case $g=\frac12$ also corresponds to the quasi-classical case (since in this case $\rho$ and $\sigma$ can be diagonalized). We will therefore assume in the rest of this subsection that $a\neq 0$ and $g\neq\frac12$.

\begin{myt}{}
\begin{theorem}\label{thmqu}[cf.~\cite{CSHO2015}]
{\rm Let $\rho,\sigma,\gamma\in\md(A)$ be three qubit states as above and suppose $a\neq 0$ and $g\neq\frac12$. Then, for $\r\neq \g$, $(\rho,\gamma)\xrightarrow{\gpc}(\sigma,\gamma)$ if and only if $(\r,\g)\succ (\s,\g)$ and
\be\label{90}
\frac{|b|^2}{|a|^2}\leq\frac{\det\begin{pmatrix}
s & 1-r \\
g & 1-g
\end{pmatrix}\det\begin{pmatrix}
r & 1-s \\
g & 1-g
\end{pmatrix}}{\left(r-g\right)^2}
\ee
For $\r=\g$, $(\rho,\gamma)\xrightarrow{\gpc}(\sigma,\gamma)$ if and only if $\s=\g$ and $|a|\geq |b|$.}
\end{theorem}
\end{myt}

\begin{remark}
If $s=g$ (but $r\neq g$) then the condition in~\eqref{90} can be simplified. Specifically, in this case we get that $(\rho,\gamma)$ can be converted to $(\sigma,\gamma)$ by GPC if and only if
\be
\frac{|b|^2}{|a|^2}\leq\det(\gamma)\;.
\ee
\end{remark}

\begin{proof}
From Lemma~\ref{lem1841} it follows that $(\rho,\gamma)$ can be converted to $(\sigma,\gamma)$ by GPC iff there exists a $2\times 2$ column stochastic matrix $P=\{p_{y|x}\}_{x,y\in\{0,1\}}$ that satisfies
$P\r=\s$, $P\g=\g$, and 
\be
\begin{pmatrix}
p_{0|0} & b/a\\
\bar{b}/\bar{a} & p_{1|1}
\end{pmatrix}\geq 0\;.
\ee 
Note that this last condition is equivalent to
\be\label{18p80}
\frac{|b|^2}{|a|^2}\leq p_{0|0}p_{1|1}\;.
\ee
The conditions $P\r=\s$ and $P\g=\g$ can be expressed as the following linear systems of equations
\be\label{linsys}
\begin{bmatrix}
r & 1-r \\
g & 1-g
\end{bmatrix}\begin{bmatrix}
p_{0|0}  \\
p_{0|1} 
\end{bmatrix}=\begin{bmatrix}
s  \\
g 
\end{bmatrix}
\ee
and
\be\label{linsys2}
\begin{bmatrix}
r & 1-r \\
g & 1-g
\end{bmatrix}\begin{bmatrix}
p_{1|0}  \\
p_{1|1} 
\end{bmatrix}=\begin{bmatrix}
1-s  \\
1-g 
\end{bmatrix}\;.
\ee
Note that the equations involving $p_{1|0}$ and $p_{1|1}$ follows trivially from the ones involving $p_{0|0}$ and $p_{0|1}$ since $P$ is column stochastic.
From Cramer's rule it then follows that for the case that $r\neq g$
\be
p_{0|0}=\frac{\det\begin{pmatrix}
s & 1-r \\
g & 1-g
\end{pmatrix}}{\det\begin{pmatrix}
r & 1-r \\
g & 1-g
\end{pmatrix}}
\quad\text{and}\quad
p_{1|1}=\frac{\det\begin{pmatrix}
r & 1-s \\
g & 1-g
\end{pmatrix}}{\det\begin{pmatrix}
r & 1-r \\
g & 1-g
\end{pmatrix}}\;.
\ee
Finally, substituting the above expression in~\eqref{18p80} gives~\eqref{90}. 

For the case that $r=g$ we also have $s=g$ (otherwise, $(\r,\g)\not\succ(\s,\g)$) and the linear system of equations in~\eqref{linsys} has a unique solution given by $p_{0|0}=p_{1|1}=1$. Therefore, in this case, \eqref{18p80} gives $|b|\leq |a|$. This completes the proof.
\end{proof}

From the remark below Theorem~\ref{thmqu} it follows that already in the qubit case, conversions under GPC have a certain type of discontinuity. To see this, consider the case $s=g$, and observe that the condition $|a|^2\det(\gamma)\geq|b|^2$ is stronger than the condition $|a|\geq |b|$ that one obtains if also $r=g$. In particular, since $\det(\gamma)\leq\frac14$, there exists an $\eps>0$ and $\rho,\sigma,\gamma\in\md(A)$ such that for any $\rho\in\mb^\eps(\sigma)$ (here $\mb^\eps(\sigma)$ is the ball of all density matrices that are $\eps$-close to $\sigma$) we have that $(\rho,\gamma)$ cannot be converted by GPC to $(\sigma,\gamma)$ unless $\rho=\sigma$ (up to free diagonal unitary).
As an explicit example, let 
\be
\sigma=\frac16\begin{pmatrix}
2 & \sqrt{2} \\
\sqrt{2} & 4
\end{pmatrix}
\quad\text{and}\quad
\gamma=\frac13\begin{pmatrix}
1 & 0 \\
0 & 2
\end{pmatrix}\;.
\ee
According to the theorem above, in this example,  $(\rho,\gamma)\xrightarrow{\gpc}(\sigma,\gamma)$ if and only if either $\rho=\sigma$ or 
\be
a\geq\frac{b}{\sqrt{\det(\gamma)}}=\frac12\;.
\ee
However, note that for sufficiently small $\eps>0$ the condition $\rho\in\mb^\eps(\sigma)$ would imply that $a$ cannot be too far away from $b=\sqrt{2}/6<1/2$. Therefore, for sufficiently small $\eps>0$ the condition $\rho\in\mb^\eps(\sigma)$ implies that $(\rho,\gamma)\xrightarrow{\gpc}(\sigma,\gamma)$ if and only if $\rho=\sigma$ (up to a free diagonal unitary).

\subsection{Approximate Single-Shot Conversions}

For the case that $\mf=\gpc$, the conversion distance given in~\eqref{32} can be expressed as
\ba
&d_\mf\big((\rho^A,\gamma^A)\to(\sigma^{A'},\gamma^{A'})\big)\eqdef\\ &\min_{\mE\in\cov(A\to A')}\left\{\frac12\left\|\sigma-\mE(\rho)\right\|_1\;:\;\gamma^{A'}=\mE\left(\gamma^A\right)\right\}.
\ea
Since the trace distance between two density matrices can be expressed as
\be
\frac12\left\|\sigma-\mE(\rho)\right\|_1=\min_{\substack{\Lambda\in\pos(A')\\\Lambda\geq \sigma-\mE(\rho)}}\tr\left[\Lambda\right]\;,
\ee
the conversion distance can be expressed as the following minimization problem
\be
d_\mf\big((\rho,\gamma)\to(\sigma,\gamma)\big)=\min \tr\left[\Lambda\right]
\ee
subject to:
\begin{enumerate}
\item $\Lambda^{A'}\geq \sigma^{A'}-\tr_A\left[J^{AA'}\left(\rho^T\otimes I^{A'}\right)\right]$.
\item $J^A=I^A$.
\item $\gamma^{A'}=\tr_A\left[J^{AA'}\left(\gamma^A\otimes I^{A'}\right)\right]$.
\item $\big[J^{AA'},e^{-i{H}^At}\otimes e^{iH^{A'}t}\big]=0$ for all $t\in\mbb{R}$.
\item $\Lambda\in\pos(A')$ and $J\in\pos(AA')$.
\end{enumerate}
The optimization problem above can be solved efficiently with semidefinite programs, however, when the output state is quasi-classical the conversion distance takes a much simpler form.


Consider the conversion distance from an arbitrary state $(\rho^A,\gamma^A)$ to a quasi-classical state $\big(\sigma^{A'},\gamma^{A'}\big)$. In this case, a channel $\mE\in\cptp(A\to A')$  is time-translation covariant if and only if $\mE\circ\mP=\mE$ since the output of the channel is classical (and therefore time-translation invariant). Therefore, in this case we get
\ba
&d_\gpc\left((\rho^A,\gamma^A)\to(\sigma^{A'},\gamma^{A'})\right)\\
&\eqdef \min_{\mE\in\cov(A\to A')}\left\{\frac12\left\|\sigma-\mE(\rho)\right\|_1\;:\;\gamma^{A'}=\mE\left(\gamma^A\right)\right\}\\
&\geq\min_{\mE\in\cov(A\to A')}\left\{\frac12\left\|\mP\big(\sigma-\mE(\rho)\big)\right\|_1\;:\;\gamma^{A'}=\mE\left(\gamma^A\right)\right\}\\
&=\min_{\mE\in\cov(A\to A')}\left\{\frac12\left\|\sigma-\mE\circ\mP(\rho)\right\|_1\;:\;\gamma^{A'}=\mE\left(\gamma^A\right)\right\}\\
&\geq d_\gp\left((\mP(\rho^A),\gamma^A)\to(\sigma^{A'},\gamma^{A'})\right)\;.
\ea
On the other hand, we have
\ba
&d_\gp\left((\mP(\rho^A),\gamma^A)\to(\sigma^{A'},\gamma^{A'})\right)\\
&=\min_{\mE\in\cptp(A\to A')}\left\{\frac12\left\|\sigma-\mE\circ\mP(\rho)\right\|_1\;:\;\gamma^{A'}=\mE\left(\gamma^A\right)\right\}\\
&\geq \min_{\mE\in\cptp(A\to A')}\left\{\frac12\left\|\sigma-\mP\circ\mE\circ\mP(\rho)\right\|_1\;:\;\gamma^{A'}=\mE\left(\gamma^A\right)\right\}\\
&\geq\min_{\mE\in\cov(A\to A')}\left\{\frac12\left\|\sigma-\mE(\rho)\right\|_1\;:\;\gamma^{A'}=\mE\left(\gamma^A\right)\right\}\\
&=d_\gpc\left((\rho^A,\gamma^A)\to(\sigma^{A'},\gamma^{A'})\right)\;.
\ea
Therefore, combining the two expressions above we get that
\ba\label{18p93}
&d_\gpc\left((\rho^A,\gamma^A)\to(\sigma^{A'},\gamma^{A'})\right)\\&= d_\gp\left((\mP(\rho^A),\gamma^A)\to(\sigma^{A'},\gamma^{A'})\right)\;.
\ea
One can then use the expression given in~\cite{R2016} for the conversion distance between two quasi-classical states. 

The above observation can be used to get an exact closed formula for the $\eps$-single-shot distillable athermality defined on any quantum athermality state $(\rho,\gamma)$ as in~\eqref{33} with $\mf=\gpc$.
Note that since the golden unit $(|0\lr 0|^X,\u_n^X)$ that appears in~\eqref{33} is quasi-classical, it cannot be used to define the single-shot cost of an arbitrary quantum athermality state since quasi-classical states cannot be converted by GPC to states with coherence between energy levels.  Therefore, in this subsection we only consider single-shot distillation of athermality.

\begin{myt}{}
\begin{theorem}\label{thmdis}
Let $\rho,\gamma\in\md(A)$ and $\eps\in[0,1]$. Then, the $\eps$-single-shot distillable athermality of the state $(\rho,\gamma)$ is given by
\be\label{50}
\distill_{\gpc}^{\eps}\left(\rho,\gamma\right)=D_{\min}^\eps\left(\mP(\rho)\big\|\gamma\right)\;.
\ee
\end{theorem}
\end{myt}
\begin{proof}
The relation in~\eqref{18p93} immediately implies that
\begin{align}
&\distill_{\gpc}^{\eps}\left(\rho,\gamma\right)=\log\sup_{0<m\in\mbb{R}}\nonumber\\
&\quad\quad\quad\quad\Big\{m\;:\;d_\gp\Big(\left(\mP(\rho),\gamma\right)\to\left(|0\lr 0|,\u_m\right)\Big)\leq\eps\Big\}\nonumber\\
&\quad\quad\quad\quad\quad\quad\quad=\distill_{\gp}^{\eps}\left(\mP(\rho),\gamma\right)\;.
\end{align}
Therefore, combining this with the first equality of~\eqref{34} yields~\eqref{50}. This completes the proof.
\end{proof}

\section{Quantum Athermality in the Asymptotic Regime}\label{sec5}

Unlike Gibbs preserving operations, both thermal operations and GPC operations cannot generate coherence between energy levels. This means that any number of copies of the golden unit  $(|0\lr 0|,\u)$ cannot be converted even to a single copy of an athermality state $(\rho^A,\gamma^A)$ that exhibits coherence between energy levels. However, it turns out that this type of irreversibility between the (infinite) cost to prepare the state $(\rho^A,\gamma^A)$ versus the finite rate at which it can be used to distill golden units of athermality, can be removed if we allow for a relatively small amount of coherence to be added to the system. 

This section is organized as follows. We start by showing that the distillable athermality of $(\rho^A,\gamma^A)$ equals to the Umegaki relative entropy (historically, a version of this result was first proved in~\cite{BHO+2013}, however, our proof has a pedagogical value as it is relatively short and is based on the single-shot regime). We then introduce a few new concepts, such as asymptotic scaling, energy spread, and sublinear athermality resources, in order to show how reversibility can be restored by appending the free operations with resources that are asymptotically negligible.

%

\subsection{Distillation of Quantum Athermality}

The formula given in~\eqref{50} for the $\eps$-single-shot distillable athermality can be used to derive the asymptotic distillable athermality. Specifically, we have
\be\label{asye2}
\distill_{\gpc}\left(\rho,\gamma\right)=\lim_{\eps\to 0^+}\limsup_{n\to\infty}\frac1nD_{\min}^\eps\left(\mP_n(\rho^{\otimes n})\big\|\gamma^{\otimes n}\right)\;.
\ee
We now use this observation, and obtain a closed formula for the RHS of the equation above.

\begin{myt}{}
\begin{theorem}\label{droat}(cf.~\cite{BHO+2013})
Let $(\rho,\gamma)$ be an athermality state of a quantum system $A$. Then,
\be
\distill_{\mf}\left(\rho,\gamma\right)=D\left(\rho\|\gamma\right)\;.
\ee
where $D$ is the Umegaki relative entropy, and $\mf$ stands for either $\tho$, $\cto$, or $\gpc$.
\end{theorem}
\end{myt}

\begin{proof}
Since $\mP(\gamma)=\gamma$ we get from~\eqref{50} and the data processing inequality (DPI) that
\ba\label{107}
\distill_{\gpc}^{\eps}\left(\rho,\gamma\right)&=D_{\min}^\eps\left(\mP(\rho)\big\|\mP(\gamma)\right)\\
\GG{DPI\to}&\leq D_{\min}^\eps\left(\rho\|\gamma\right)\;,
\ea
and consequently
\ba
\distill_{\gpc}\left(\rho,\gamma\right)&\leq\lim_{\eps\to 0^+}\limsup_{n\to\infty}\frac1nD_{\min}^\eps\left(\rho^{\otimes n}\big\|\gamma^{\otimes n}\right)\\
&=D(\rho\|\gamma)\;,
\ea
where the equality follows from the quantum Stein's lemma. Since $\tho$ and $\cto$ are subsets of $\gpc$ the above inequality also holds if we replace $\gpc$ with $\tho$ or $\cto$.

To get the opposite inequality, fix $k\in\mbb{N}$ and apply the pinching channel $\mP_k\in\cto(A^k\to A^k)$ on $k$ copies of $\rho$. The resulting state, $\sigma_k\eqdef\mP_k(\rho^{\otimes k})$, is quasi-classical. Now, from~\eqref{39} we get
\ba
\distill_{\tho}\left(\rho,\gamma\right)&\geq \frac1k\distill_{\tho}\left(\rho^{\otimes k},\gamma^{\otimes k}\right)\\
&\geq \frac1k\distill_{\tho}\left(\sigma_k,\gamma^{\otimes k}\right)\;,
\ea
where in the second inequality we used the fact that the pinching channel is a thermal operation (see Lemma~\ref{lem2}) so that by definition, $\sigma_k$ cannot have a higher distillable rate than $\rho^{\otimes k}$. Since $\sigma_k$ is a quasi-classical state we have that $\distill_{\tho}\left(\sigma_k,\gamma^{\otimes k}\right)$ equals $D\left(\sigma_k\big\|\gamma^{\otimes k}\right)$. Therefore, the equation above gives
\be
\distill_{\tho}\left(\rho,\gamma\right)
\geq\frac1kD\left(\mP_k\left(\rho^{\otimes k}\right)\big\|\gamma^{\otimes k}\right)\;.
\ee
Now, since the above inequality holds for all $k\in\mbb{N}$ we conclude that
\ba
\distill_{\tho}\left(\rho,\gamma\right)&\geq\limsup_{k\to\infty}\frac1kD\left(\mP_k\left(\rho^{\otimes k}\right)\big\|\gamma^{\otimes k}\right)\\
\GG{\eqref{18p27}\to}&=D(\rho\|\gamma)\;.
\ea
This completes the proof.
%
%
%
%
\end{proof}


\subsection{Scaling of Time-Translation Asymmetry}

In this subsection we show that the coherence of $n$ copies of a states grows logarithmically with $n$.
Let $A$ be a physical system with Hamiltonian $H\in\pos(A)$ and a state $\psi\in\pure(A)$ given by 
$$
H^A=\sum_{x=1}^ma_x|x\lr x|\quad\text{and}\quad|\psi\ra=\sum_{x=1}^m\sqrt{p_x}|x\ra\;,
$$ 
where $m=|A|$. 
For any $n\in\mbb{N}$, the state $\psi^{\otimes n}$ has the form 
\ba
|\psi\ra^{\otimes n}&=\sum_{x^n\in[m]^n}\sqrt{p_{x^n}}|x^n\ra\\
&=\sum_{x^n\in[m]^n}2^{-\frac n2\big(H(\t(x^n))+D\left(\t(x^n)\|\p\right)\big)}|x^n\ra
\ea
where we used~\eqref{tpxn}. For any type $\t\in\type(n,m)$ define
\be\label{124}
|\t\ra^{A^n}\eqdef \frac{1}{{n\choose nt_1,...,nt_m}^{1/2}}\sum_{x^n\in x^n(\t)}|x^n\ra\;,
\ee
where the sum runs over all sequences $x^n$ of the same type $\t$.
With the above notations
\be\label{psi15}
|\psi\ra^{\otimes n}=\sum_{\t\in\type(n,m)}\sqrt{r_{\t,n}}|\t\ra^{A^n}
\ee
where
\be
r_{\t,n}\eqdef |x^n(\t)|2^{-n\big(H(\t)+D\left(\t\|\p\right)\big)}\;.
\ee
Note that the vectors $|\t\ra^{A^n}$ are eigenvectors of the Hamiltonian of system $A^n$. Specifically,
\be\label{15p106}
H^{\otimes n}|\t\ra^{A^n}=n\sum_{x=1}^mt_xa_x|\t\ra^{A^n}\;,
\ee
so that the energy in the state $|\t\ra^{A^n}$ is $n$ times the average energy with respect to the type $\t$. In the generic case,  the energy eigenvalues $\{a_1,...,a_m\}$ are rationally independent; i.e. for any set of $m$ integers $\ell_1,...,\ell_m\in\mbb{Z}$ we have 
\be
\ell_1a_1+\cdots+\ell_ma_m=0\;\iff \;\ell_1=\cdots=\ell_m=0\;.
\ee
Under this mild assumption (which we will \emph{not} assume, but still worth mentioning), each type in $\type(n,m)$ determines uniquely the energy of the system. 

Since each $|\t\ra^{A^n}$ is an energy eigenstate, it follows from~\eqref{psi15} that we can always write $|\psi\ra^{\otimes n}$ as a linear combination of $|\type(n,m)|\leq (n+1)^m$ energy eigenstates.  In other words, the coherence of $|\psi\ra^{\otimes n}$ can be compressed into an $(n+1)^m$ dimensional vector (which is polynomial in $n$). 

The observation above has the following consequence. In Corollary~\ref{puretomixgpc} we saw that for any mixed state in $\md(A)$ there exists a pure state in $\pure(A)$ that can be converted to it by GPC. Combining this with the above observation implies that the pure state coherence cost of preparing a state $\rho^{\otimes n}\in\md(A^n)$ cannot exceed $m\log(n+1)$, and specifically we have
\be\label{108}
C\left(\rho^{\otimes n}\right)\leq m\log(n+1)
\ee
where $C$ is the coherence measure defined in~\eqref{coherence}.
Therefore, the asymmetry cost rate, i.e. coherence cost per copy of $\rho$, cannot exceed $m\frac{\log(n+1)}{n}$ which goes to zero as $n\to\infty$. On the other hand, the non-uniformity cost (i.e. quasi-classical athermality cost) does \emph{not} go to zero in the asymptotic limit since the energy of $\rho^{\otimes n}$ grows linearly with $n$.

To summarize, athermality comprise of two types of resources, namely nonuniformity (also known as informational non-equilibrium) and time-translation asymmetry (or coherence in short). Therefore, the athermality asymptotic cost of an athermality state also comprise of two types, nonuniformity cost and coherence cost. The latter however goes to zero in the asymptotic limit, and therefore needs to be rescaled differently. This makes the QRT of athermality very subtle, and consequently several important questions in the theory are still open (see Sec.~\ref{sec6}).

\subsection{The Energy Spread}

The energy spread of a given pure state $\psi\in\pure(A)$ is defined as the difference between the maximal and minimal energies that appear when writing $\psi$ as a superposition of energy eigenvectors.
In the discussion above we saw that $n$ copies of a state $\psi\in\pure(A)$ can be expressed as a linear combination of no more that $(n+1)^m$ energy eigenvectors. Among these energy eigenvectors are 
the zero energy eigenvector (corresponding to the type $\t=(1,0,...,0)^T$) and the maximal energy eigenvector (corresponding to the type $\t=(0,...,0,1)^T$). Therefore, since the energy in the decomposition~\eqref{psi15} spreads from zero to $na_m$ (where $a_m$ is the maximal energy of a single copy of system $A$), we conclude that the energy spread of $\psi^{\otimes n}$ is $na_m$.

The energy spread can be reduced drastically if one allows for a small deviation from the state $\psi^{\otimes n}$. Explicitly, 
for any $\eps>0$ we can split $|\psi^{\otimes n}\ra$ into two parts
\be\label{splitp}
|\psi\ra^{\otimes n}=\sum_{\t\in\ms_{n,\eps}}\sqrt{r_{\t,n}}|\t\ra^{A^n}+\sum_{\t\in\ms_{n,\eps}^c}\sqrt{r_{\t,n}}|\t\ra^{A^n}\;,
\ee
where
\be\label{msn}
\ms_{n,\eps}\eqdef\big\{\t\in\type(n,m)\;:\;\frac12\|\t-\p\|_1\leq\eps\big\}\;,
\ee
and $\ms_{n,\eps}^c$ is the complement of $\ms_{n,\eps}$ in $\type(n,m)$.
By substituing the bounds in~\eqref{738} into the definition of the coefficients $r_{\t,n}$ we get that
\be
\frac1{(n+1)^m}2^{-nD(\t\|\p)}\leq r_{\t,n}\leq 2^{-nD(\t\|\p)}\;.
\ee
Therefore, the fidelity of $|\psi^{\otimes n}\ra$ with the second term on the RHS of~\eqref{splitp} is given by
\ba\label{pinsker}
\sum_{\t\in\ms_{n,\eps}^c}r_{\t,n}&\leq \sum_{\t\in\ms_{n,\eps}^c}2^{-nD(\t\|\p)}\\
\GG{Pinsker's\; inequality\to}&\leq \sum_{\t\in\ms_{n,\eps}^c}2^{-2n\eps^2}\\
&\leq 2^{-2n\eps^2}\big|\type(n,m)\big|\\
\GG{\eqref{typeb}\to}&\leq 2^{-2n\eps^2}(n+1)^m\\
&\xrightarrow{n\to\infty}0\;.
\ea
That is, for any $\eps>0$ and sufficiently large $n$, the state $|\psi\ra^{\otimes n}$ can be made arbitrarily close to the state
\be\label{18p112}
|\psi_\eps^n\ra\eqdef\frac1{\sqrt{\nu_\eps}}\sum_{\t\in\ms_{n,\eps}}\sqrt{r_{\t,n}}|\t\ra^{A^n}
\ee
where
$
\nu_\eps\eqdef\sum_{\t\in\ms_{n,\eps}}r_{\t,n}
$ is the normalization factor.
Now, from~\eqref{15p106}  the energy of any state $|\t\ra^{A^n}$ with type $\t\in\ms_{n,\eps}$ is $\mu_\t\eqdef n\sum_{x=1}^mt_xa_x$.
Expressing $\t=\p+\r$ we get that $\frac12\|\r\|_1\leq\eps$ and
\be
|\mu_\t-\mu_\p|\leq n\sum_{x=1}^ma_x|t_x-p_x|=n\sum_{x=1}^ma_x|r_x|\leq 2n\eps\sum_{x=1}^ma_x.
\ee
Therefore, for any two types $\t,\t'\in\type(n,m)$ that are $\eps$-close to $\p$ we have
\be\label{18p107}
|\mu_\t-\mu_{\t'}|\leq 4n\eps\sum_{x=1}^ma_x
\ee
In other words, the energy spread of the state $|\psi^n_\eps\ra$ is no greater than $4n\eps\sum_{x=1}^ma_x$.

Note that by taking $\eps>0$ sufficiently small we can make the energy spread $4n\eps\sum_{x=1}^ma_x$ much smaller than $na_m$. However, we still get that the energy spread of $\psi_\eps^n$ is linear in $n$. We show now that there exist states in $\pure(A^n)$ that are very close to $\psi^{\otimes n}$ but with energy spread that is sublinear in $n$.

\begin{lemma}\label{lem:nchin}
Let $\psi\in\pure(A)$ and $\alpha\in\left(1/2,1\right)$. Then, there exists a sequence of pure state $\{\chi_n\}_{n\in\mbb{N}}$ in $\pure(A^n)$ with the following properties:
\begin{enumerate}
\item The limit
\be\label{18p108}
\lim_{n\to\infty}\left\|\psi^{\otimes n}-\chi_n\right\|_1=0\;.
\ee
\item The energy spread of $\chi_n$ is no more than $4n^{\alpha}\sum_{x=1}^ma_x$.
\end{enumerate}  
\end{lemma}

\begin{remark}
Note that due to the inequality $\type(n,m)\leq (n+1)^m$ it follows that any pure state in $A^n$, including $|\chi_n\ra$, can be expressed as a linear combination of no more than $(n+1)^m$ energy eigenstates.
\end{remark}

\begin{proof}
Set $\eps_n\eqdef n^{\alpha-1}$. Since $\alpha\in\left(\frac12,1\right)$ we have $\lim_{n\to\infty}\eps_n=0$ and $\lim_{n\to\infty}n\eps_n^2=\infty$. The latter implies that if we replace $\eps$ in~\eqref{pinsker} with $\eps_n$ we still get the zero limit of~\eqref{pinsker}. Hence, the pure state $\chi_n\eqdef\psi_{\eps_n}^n$ satisfies~\eqref{18p108}.
 Finally, from~\eqref{18p107} we get that the energy spread of $\chi_n$ cannot exceed 
\be
4n\eps_n\sum_{x=1}^ma_x=4n^{\alpha}\sum_{x=1}^ma_x\;.
\ee
This completes the proof.
\end{proof}


\subsection{Sublinear Athermality Resources}

The lemma above asserts that the state $\psi^{\otimes n}$ is very close to a state $\chi_n$, whose energy spread is sublinear in $n$. However,  the average energy $\la\chi_n|H^{\otimes n}|\chi_n\ra$ grows linearly in $n$. This is consistent with our next assumption that systems whose energy grows sub-linearly in $n$ can be viewed as asymptotically negligible resources.  
%

\begin{definition}
A sublinear athermality resource (SLAR) is a sequence of quantum athermality systems $\{R_n\}_{n\in\mbb{N}}$, such that $|R_n|$ grows polynomially with $n$, and there exists two constants independent of $n$, $0\leq \alpha<1$ and $c>0$, such that
\be
\left\|H^{R_n}\right\|_{\infty}\leq cn^\alpha\quad\quad\forall\;n\in\mbb{N}\;.
\ee
\end{definition}

The key assumption in the definition above is that the energy of systems $R_n$ grows sublinearly with $n$. Therefore, in the asymptotic limit in which $n\to\infty$ the resourcefulness of any sequence of athermality states $\big\{(\omega^{R_{n}},\gamma^{R_n})\big\}_{n\in\mbb{N}}$ becomes negligible relative to the resourcefulness of $n$ copies of the golden unit $(|0\lr 0|^A,\u^A)$. Specifically, in Appendix~\ref{appendixb} we show that the distillation rate of athermality as given in Theorem~\ref{droat} does not change if we replace CTO (or GPC) by CTO+SLAR (or GPC+SLAR). While this  small amount of an athermality resource does not change the distillation rate, we will see now that it does change the cost rate and thereby sufficient to restore reversibility.

\subsection{Cost of Pure States}

For any athermality system $R$ (i.e. system $R$ has a well define Hamiltonian $H^R$ and a Gibbs state $\gamma^R$) we define the $R$-assisted conversion distance of one athermality state $(\rho^A,\gamma^A)$ to another athermality state $(\sigma^B,\gamma^B)$ as
\ba\label{129}
&d_{\cto}^R\Big(\left(\rho^A,\gamma^A\right)\to \left(\sigma^B,\gamma^B\right)\Big)\\
&\eqdef\inf_{\omega\in\md(R)}
d_{\cto}\Big(\left(\rho^A\otimes\omega^{R},\gamma^{AR}\right)\to \left(\sigma^B,\gamma^B\right)\Big)\;.
\ea 
That is, $d_{\cto}^R$ is the smallest distance that $(\rho^A,\gamma^A)$ can be reached by CTO to $(\sigma^B,\gamma^B)$ with the help of a system $R$, whose Hamiltonian $H^R$ (or equivalently its Gibbs state $\gamma^R$) is fixed. 
With this at hand, we define the $R$-assisted $\eps$-cost of  $(\rho^A,\gamma^A)$ to be
\begin{align*}
&\cost_\cto^{\eps,R}\left(\rho^A,\gamma^A\right)\eqdef\\
&\log\inf_{0<m\in\mbb{R}}\left\{m\;:d_{\cto}^R\left((|0\lr 0|^X,\u_m^X)\to (\rho^A,\gamma^A)\right)\leq\eps\right\}\;.
\end{align*}
The type of free operations that we consider here are CTO assisted with SLAR. We therefore set in this subsection $\mf$ to be CTO+SLAR. Using the definitions above, we define the asymptotic cost of a state $(\rho^A,\gamma^A)$ under $\mf$ as
\ba
&\cost_\mf\left(\rho^A,\gamma^A\right)\eqdef\\
&\inf_{\{R_n\}}\lim_{\eps\to 0^+}\liminf_{n\to\infty}\frac1n\cost_\cto^{\eps,{R_n}}\left(\rho^{\otimes n},\gamma^{\otimes n}\right)
\ea
where the infimum is over all SLARs $\{R_n\}_{n\in\mbb{N}}$.
%
%

\begin{myt}{}
\begin{theorem}\label{purecostf}
Let $(\psi^A,\gamma^A)$ be an athermality state with $\psi\in\pure(A)$. Then,
\be
\cost_\mf\left(\psi^A,\gamma^A\right)=D\left(\psi^A\big\|\gamma^A\right)\;,
\ee
where $D$ is the Umegaki relative entropy.
\end{theorem}
\end{myt}

\begin{proof}
Let $\ms_{n,\eps}$ be be the set of types given in~\eqref{msn} and set $\ms_n\eqdef\ms_{n,\eps_n}$ with $\eps_n\eqdef n^{\alpha-1}$. Let also $\{\chi_n\}_{n\in\mbb{N}}$ be the sequence of pure states that satisfies all the properties outlined in Lemma~\ref{lem:nchin}. In particular, each $\chi_n$ is very close to $\psi^{\otimes n}$ (for $n$ sufficiently large) and the energy spread of $\chi_n$ is given by $4n^\alpha\sum_{x=1}^ma_x$ for some $\alpha\in(\frac12,1)$. Recall that each $\chi_n$ has the form (cf.~\eqref{18p112})
\be\label{chisuper}
|\chi_n\ra=\sum_{\t\in\ms_n}\sqrt{q_{\t}}|\t\ra^{A^n}\;,
\ee
where $\{q_\t\}$ are some coefficients in $\mbb{R}_+$ (that form a probability distribution over the set of types in $\ms_n$).
Let $k_n\eqdef|\ms_n|$ be the number of terms in the superposition above (hence $k_n\leq (n+1)^m$), and let $\{\mu_j\}_{j=1}^{k_n}$ be the set of all energy eigenvalues of the Hamiltonian $H^{\otimes n}$ that corresponds to all the energy eigenvectors $\{|\t\ra^{A^n}\}_{\t\in\ms_n}$. That is, each $j\in[k_n]$ corresponds exactly to one type $\t$ that appears in the superposition~\eqref{chisuper}. Although the energies eigenvalues $\{\mu_j\}$ depend also on $n$, we did not add a subscript $n$ to ease on the notations. W.l.o.g.\ we also assume that $\mu_1\leq\cdots\leq\mu_{k_n}$, so that its energy spread $\mu_{k_n}-\mu_1\leq 4n^\alpha\sum_{x=1}^ma_x$ (see Lemma~\ref{lem:nchin}). We will also denote by $\t^{\min,n}$ the type in $\ms_n$ that corresponds to the smallest energy $\mu_1$.

Let $R_n$ be a $k_n$-dimensional quantum (reference) system whose Hamiltonian has non-degenerate spectrum given by
\be
H^{R_n}=\sum_{j=1}^{k_n}(\mu_j-\mu_1)|j\lr j|^{R_n}\;.
\ee
Note that the Hamiltonian $H^{R_n}$ has the same eigenvalues as the energies that appear in $\chi_n$ shifted by $\mu_1$. Set $\lambda_j\eqdef\mu_j-\mu_1$ to be the $j$th eigenvalue of $H^{R_n}$, and observe that 
\be
0=\lambda_1\leq \lambda_2\leq\cdots\leq \lambda_k\leq 2n^{\alpha}\sum_{x=1}^ma_x\;.
\ee 
Let $z^n\in[m]^n$ be a sequence of type $\t^{\min,n}$ so that $H^{\otimes n}|z^n\ra^{A^n}=\mu_1|z^n\ra^{A^n}$. Let also
\be
|\phi^{R_n}\ra\eqdef \sum_{j=1}^{k}\sqrt{q_j}|j\ra^{R_n}
\ee
where $q_{j}\eqdef q_{\t}$ with $\t$ being the type that corresponds to $j$; i.e.\ $\t$ is the type satisfying $H^{\otimes n}|\t\ra^{A^n}=\mu_j|\t\ra^{A^n}$. By construction, $\{R_n\}_{n\in\mbb{N}}$ is a SLAR, and the pure state
\be\label{136}
\phi^{R_n}\otimes|z^n\lr z^n|^{A^n}
\ee
has the exact same energy distribution as the pure state
\be\label{137}
|1\lr 1|^{R_n}\otimes \chi_n^{A^n}
\ee
(recall that $|1\ra^{R_n}$ corresponds to the zero energy of system $R_n$). Hence, the above two states are equivalent resources and can be converted from one to the other by reversible thermal operations (i.e.\ an energy preserving unitary). We now use this resource equivalency to compute the cost of $\psi^{\otimes n}$ in terms of the cost of the quasi-classical state $|z^n\lr z^n|$.

Let $\eps\in(0,1/2)$ and let $n\in\mbb{N}$
 be sufficiently large such that $\psi^{\otimes n}$ is $\eps$-close to $\chi_n$. Therefore, any positive real number $0<m\in\mbb{R}$ that satisfies 
 \be
 d_{\cto}^{R_n}\left((|0\lr 0|^X,\u_m^X)\to (\chi_n^{A^n},\gamma^{A^n})\right)\leq\eps
 \ee
also satisfies 
 \be
 d_{\cto}^{R_n}\left((|0\lr 0|^X,\u_m^X)\to (\psi^{\otimes n},\gamma^{A^n})\right)\leq 2\eps\;.
 \ee
 In particular, this means that
\be
\cost_{\cto}^{2\eps,R_n}\left(\psi^{\otimes n},\gamma^{A^n}\right)\leq\cost_{\cto}^{\eps,R_n}\left(\chi_n^{A^n},\gamma^{A^n}\right)\;.
\ee
We will therefore focus now on bounding the expression on the RHS above.

By adding the resource $(|1\lr 1|^{R_n},\gamma^{R_n})$ we can only increase the cost. Therefore,
\ba
\cost_{\cto}^{\eps,R_n}&\left(\chi_n^{A^n},\gamma^{A^n}\right)\\
&\leq\cost_{\cto}^{\eps,R_n}\left(|1\lr 1|^{R_n}\otimes\chi_n^{A^n},\gamma^{R_nA^n}\right)\\
&=\cost_{\cto}^{\eps,R_n}\left(\phi^{R_n}\otimes|z^n\lr z^n|^{A^n},\gamma^{R_nA^n}\right)\\
&\leq \cost_{\cto}^{\eps}\left(|z^n\lr z^n|^{A^n},\gamma^{A^n}\right)\\
&=D_{\max}^\eps\left(|z^n\lr z^n|^{A^n}\big\|\gamma^{A^n}\right)\\
&\leq D_{\max}\left(|z^n\lr z^n|^{A^n}\big\|\gamma^{A^n}\right)\;,
\ea
where in the first equality we used the resource equivalency between the athermality states in~\eqref{136} and~\eqref{137}. In the second inequality we used the fact that the cost of $|z^n\lr z^n|$ without the assistance of $R_n$ cannot be smaller than the cost of $\phi^{R_n}\otimes|z^n\lr z^n|$ with the assistance of $R_n$, since the latter is defined in terms of an infimum over all states $\omega\in\md(R_n)$ (cf.~\eqref{129}). In the second equality we used the second relation of~\eqref{34} combined with the fact that in the quasi-classical regime GPO has the same conversion power as CTO (see~\eqref{19}). Finally, in the last inequality we used the fact that $D_{\max}$ is always no smaller than its smoothed version. Combining this with the previous equation and with the definition of $\cost_{\mf}\left(\psi^{A},\gamma^{A}\right)$, we conclude that
\be
\cost_{\mf}\left(\psi^{A},\gamma^{A}\right)\leq \liminf_{n\to\infty}\frac1nD_{\max}\left(|z^n\lr z^n|^{A^n}\big\|\gamma^{A^n}\right)\;.
\ee 
Now, observe that
\ba\label{143}
D_{\max}\left(|z^n\lr z^n|^{A_n}\big\|\gamma^{A^n}\right)&=-\log\big\la z^n\big|\gamma^{A^n}\big|z^n\big\ra\\
\Gg{\gamma^{A^n}=\left(\gamma^A\right)^{\otimes n}\to}&=-\sum_{x=1}^mnt_x^{n,\min}\log\la x|\gamma^{A}|x\ra\;,
\ea
where in the last equality we used the fact that the sequence $z^n$ has a type $\t^{\min,n}$.
Combining this with the previous equation we conclude that
\ba\label{18p127}
\cost_{\mf}\left(\psi^{A},\gamma^{A}\right)&\leq-\lim_{n\to\infty}\sum_{x=1}^mt_x^{n,\min}\log\la x|\gamma^{A}|x\ra\\
&=-\sum_{x=1}^{m}p_x\log\la x|\gamma^{A}|x\ra\\
&=D(\psi\|\gamma)\;,
\ea
where we used the fact that $\t^{\min,n}\in\ms_n$ so that $\t^{\min,n}$ is $\eps_n$-close to $\p$. Therefore, since $\lim_{n\to\infty}\eps_n=0$ we have $\lim_{n\to\infty}\t^{\min,n}=\p$. This completes the proof.
\end{proof}

\section{Discussion and Outlook}\label{sec6}

Quantum athermality can be viewed as a  composite resource consisting of non-uniformity and quantum coherence. While the study of non-uniformity is well understood, the role coherence plays in quantum thermodynamics is far less understood. In this paper, we first developed  the resource theory of time-translation asymmetry which is the type of quantum coherence appearing in thermodynamics. Remarkably, we were able to find (Theorem~\ref{ttcs}) a relatively simple criterion, determining if there exists a time-translation covariant channel between two given quantum states. We restricted our attention to Hamiltonians with non-degenerate Bohr spectrum as almost all Hamiltonians have such a spectrum. However, it is worth noting that some important Hamiltonians, such as the Hamiltonian of the harmonic oscillator, do not have such a spectrum. For such Hamiltonians, some of our results do not apply, although significant progress has been made recently in this direction~\cite{Marvian2021,KH2022,KT2022}.

We used the resource theory of time-translation asymmetry to develop the theory of quantum athermality in the single-shot regime. We considered three types of free operations: thermal operations (TO), closed thermal operations (CTO), and Gibbs preserving covariant operations (GPC). In particular, Theorem~\ref{thm31} demonstrated in a rigorous way that two athermality states $\rho,\sigma\in\md(A)$ with the same diagonal elements have the same non-uniformity content, and can only be different in their coherence content. 

In this respect, it would be interesting to know if the same result holds also for CTO. 
One of the long-standing open problems in the resource-theoretic approach to quantum thermodynamics is whether under CTO, quantum athermality comprise of just non-uniformity and coherence. That is, since GPC is a larger set of operations than CTO it could well be that some interconversions between two athermality resources is possible under GPC operations but not under CTO. If this is the case, it would mean that quantum athermality contains another type of resource that is not captured solely by coherence and nonuniformity.

In the asymptotic regime, however, GPC does not provide any advantage over CTO. Both sets of operations lead to the same distillable rate given in terms of the Umegaki relative entropy (see Theorem~\ref{droat}). Since coherence is needed to create athermality states that are not quasi-classical, the cost rate of a non-quasi-classical state diverges. To get a meaningful result, we followed the idea of~\cite{BHO+2013} to borrow a small amount of coherence, and showed that, for pure states, with the assistance of an asymptotically negligible quantum athermality, we can restore into the fully quantum domain, the reversibility that exists in the quasi-classical regime.

We defined asymptotically negligible resources as sequences $\{R_n\}_{n\in\mbb{N}}$ whose maximal energy grows sublinearly with $n$. The intuition behind this definition is that the energy of $n$ copies of a system $A$ grows linearly with $n$ so that for sufficiently large $n$, the energy of $A^n$ is much larger than that of $R_n$. Indeed, such a sublinear athermality resource (SLAR) cannot increase the distillable athermality (see Appendix~\ref{appendixb}).

In Theorem~\ref{purecostf} we showed that the cost rate of a \emph{pure} athermality resource, under $\mf\eqdef\cto+{\rm SLAR}$, is given by the Umegaki relative entropy. Moreover, in Theorem~\ref{droat} and Appendix~\ref{appendixb} we showed that the distillable rate under $\mf$ of \emph{any} athermality resource is given by the Umegaki relative entropy. When combining these two results together we conclude that the rate of converting (by $\mf$) many copies of a mixed state $(\rho^A,\gamma^A)$ to many copies of a pure state $(\psi^{B},\gamma^{B})$ is given by
\be
{\rm Rate}_\mf\Big(\left(\rho^A,\gamma^A\right)\to \left(\psi^{B},\gamma^{B}\right)\Big)=\frac{D(\rho^A\|\gamma^A)}{D(\psi^{B}\|\gamma^{B})}\;.
\ee
For the specific case that also $\rho^A$ is pure, the above formula indicates that the resource theory of pure athermality is reversible under $\mf$. For the mixed-state case the problem is still open.

As discussed above, under GPC and CTO, coherence among energy level is a resource that cannot be measured by the golden unit $(|0\lr 0|^X,\u_m^X)$ introduced in~\cite{WW2019} (see~\eqref{334}) for athermality under GPO. The reason is that this golden unit is quasi-classical, and it cannot be converted by GPC (or CTO) to any athermality state that is not quasi-classical (even if we take $m=\infty$).  For this reason, one has to specify another golden unit that quantifies the coherence content of quantum athermality. We discuss now a candidate of such a golden unit.

For a given athermality state $(\rho,\gamma)$ we can interpret the state $(\mP(\rho),\gamma)$ as the non-uniformity contained in $(\rho,\gamma)$. If fact, we saw in Theorem~\ref{thmdis} that for any $\eps>0$, $\distill_{\gpc}^{\eps}\left(\rho,\gamma\right)=\distill_{\gp}^{\eps}\left(\mP(\rho),\gamma\right)$, which supports this assertion. It is somewhat less clear how to characterize or isolate the time-translation asymmetry contained in $(\rho,\gamma)$. 

Consider an athermality state $(\sigma,\gamma)$ with the property that $\mP(\sigma)=\gamma$. Such an athermality state has zero nonuniformity, and consequently it contains only time-translation asymmetry. We can therefore call such states purely-coherent athermality states.
In Corollary~\ref{puretomixgpc} we saw that the purely-coherent athermality state $(\psi_\gamma,\gamma)$,  given by
\be
|\psi_\gamma\ra\eqdef\sum_{x=1}^m\sqrt{g_x}|x\ra\quad\text{and}\quad\gamma=\sum_{x=1}^mg_x|x\lr x|\;,
\ee
can be converted to any other purely-athermality state of the form $(\rho,\gamma)$, where $\rho\in\md(A)$ has the same diagonal as $\gamma$. Therefore, the athermality state $(\psi_\gamma,\gamma)$ can be taken to be the golden unit for the coherence content of quantum athermality. Note however that unlike the golden unit $(|0\lr 0|,\u_m)$ used for the nonuniformity content of athermality, $(\psi_\gamma,\gamma)$ depends explicitly on the Hamiltonian. 

With this golden unit, we can now ask what is the coherence cost of an athermality state $(\rho,\gamma)$. To compute the asymptotic cost of preparing many copies, say $n$, of a given athermality state $(\rho,\gamma)$ we can minimize the integers $m,k$ for which the conversion
\be
\big(\psi_\gamma,\gamma\big)^{\otimes k}\otimes\big(|0\lr 0|,\u\big)^{\otimes m}\;\xrightarrow{\gpc}\;\big(\rho,\gamma\big)^{\otimes n}
\ee
is possible with a small error that vanishes in the limit $n\to\infty$. We leave the investigation along these lines for future work.

\begin{acknowledgments}
The author would like to thank David Jennings, Thomas Theurer, and Marco Tomamichel for useful discussions. The author also thanks both Thomas Theurer and Ludovico Lami for extremely useful comments on the first draft of the paper. The new, shorter proof, for Corollary~\ref{puretomix} that appear in the current version is due to Ludovico Lami. 
The author acknowledge support from the Natural Sciences and Engineering Research Council of Canada (NSERC).
\end{acknowledgments}

\bibliographystyle{apsrev4-2}
\bibliography{QRT}

\begin{appendix}

\section{Possible gaps in the original proof of~\cite{BHO+2013}}

The proof given in~\cite{BHO+2013} seems to have several gaps. Here we point out one such gap, and discuss an implicit assumption made in~\cite{BHO+2013}.

In $(35)$ and $(37)$ of their Supplemental Material (SM) the authors of~\cite{BHO+2013} consider two states
\be
\rho^{\otimes n}=\sum_{k,g}p_k|\Psi_{k,g}\lr \Psi_{k,g}|\;,\;\;\rho_n=\sum p_k|t_k,s_g\lr t_k,s_g|\;,
\ee
where for simplicity, the authors consider rank 2 state \be\rho=p|\phi_1\lr\phi_1|+(1-p)|\phi_2\lr \phi_2|\;.
\ee
In the first step of their protocol, one first create the diagonal state $\rho_n$ which has the same spectrum as $\rho^{\otimes n}$. Since $\rho_n$ is diagonal, its eigenvectors $\{|t_k,s_g\ra\}$ depend only on the Gibbs state $\gamma^{\otimes n}$.  The authors of~\cite{BHO+2013} do not specify in (37) (of their SM) the range of $k$, but from (45) in the SM it becomes clear that $k\in{\rm Typ}_\rho\eqdef[np-\sqrt{n},np+\sqrt{n}]$   (see the sentence above (38) in the SM of~\cite{BHO+2013} for the definition ${\rm Typ}_\rho\eqdef[np-\sqrt{n},np+\sqrt{n}]$).

The author then move to claim that
 ``From the result of the previous section it is not hard to see that this [i.e. the cost of preparing $\rho_n$] can be done at a rate given by the relative entropy distance of $\rho$ to the Gibbs state, since in the limit of many copies, the regularized relative entropy distance is the same".
 However, there exists a simple argument why, in general, the cost of preparing $\rho_n$ is not equal to $D(\rho\|\gamma)$, where $\gamma$ is the Gibbs state. 
 
The argument goes as follows. Consider  the two states $\rho$ and $\sigma\eqdef V\rho V^\dag$, where $V$ is some unitary matrix. Since the eigenvalues of both $\rho$ and $\sigma$ are $p$ and $1-p$, it follows that the construction of $\rho_n$ would be exactly the same whether our initial state is $\rho$ or whether it is $\sigma$. This is because $\rho_n$ does not depend explicitly on the eigenvectors of $\rho$ (only the eigenvalues). However, clearly, there exists a unitary $V$ such that 
 \be
 D(V\rho V^*\|\gamma)\neq D(\rho\|\gamma)\;.
 \ee
Since $\rho_n$ as defined above would be the same for both $\rho$ and $\sigma\eqdef V\rho V^*$ the cost rate of preparing $\rho_n$ cannot be equal to $D(\rho\|\gamma)$.

Another issue with the proof in~\cite{BHO+2013} is that the matrix in $(31)$ of the SM of~\cite{BHO+2013} is not a unitary matrix as claimed.
Indeed, by direct calculation
\ba
&U^{\rm inv}(U^{\rm inv})^\dag=\\
&\sum_{i,i',j}u_{ij}\bar{u}_{i'j}|E_i\lr E_{i'}|\otimes|h-E_i+E_j\lr h-E_{i'}+E_j|\\
&\neq I
\ea
Perhaps the intention of the authors of~\cite{BHO+2013} is to include a sum over $h$ in the definition of $U^{\rm inv}$, and allowing $h$ to go from $-\infty$ to $+\infty$ so that
\ba
&U^{\rm inv}(U^{\rm inv})^\dag=\\
&\sum_{i,i',j}u_{ij}\bar{u}_{i'j}|E_i\lr E_{i'}|\otimes\sum_{h=-\infty}^\infty |h-E_i+E_j\lr h-E_{i'}+E_j|\\
&=\sum_{i,i',j}u_{ij}\bar{u}_{i'j}|E_i\lr E_{i'}|\otimes\sum_{h=-\infty}^\infty |h-E_i\lr h-E_{i'}|\\
&=\sum_{i,i'}\delta_{i,i'}|E_i\lr E_{i'}|\otimes\sum_{h=-\infty}^\infty |h-E_i\lr h-E_{i'}|\\
&= I\;.
\ea
Note that one has to allow for the Hamiltonian of the reference system to have an unbounded negative spectrum.  Such Hamiltonians are known to lead to instabilities of the physical system, and occur for example in relation to spin-statistics theorem. However, we point out, that in the present paper the author only assumes ancillary systems of finite dimensions and with Hamiltonians whose spectrum is non-negative (i.e. bounded from below).

\section{Distillation under GPC+SLAR}\label{appendixb}

\begin{lemma}
Let $(\rho,\gamma)$ be an athermality state, and let $\mf\eqdef\gpc+{\rm SLAR}$. Then,
\be
\distill_{\mf}\left(\rho,\gamma\right)=D\left(\rho\|\gamma\right)\;.
\ee
\end{lemma}

\begin{proof}
Let $\{(\omega^{R_n},\gamma^{R_n})\}_{n\in\mbb{N}}$ an an SLAR and observe that from~\eqref{107} it follows that for any $n\in\mbb{N}$ and any $\eps\in(0,1)$
\ba
&\distill_{\gpc}^{\eps}\left(\rho^{\otimes n}\otimes\omega^{R_n},\gamma^{A^n}\otimes\gamma^{R_n}\right)\\
&\quad\quad\quad\leq D_{\min}^\eps\left(\rho^{\otimes n}\otimes\omega^{R_n}\big\|\gamma^{A^n}\otimes\gamma^{R_n}\right)\\
&{\GG{\eqref{47}\to}}\leq D_{\min}^\eps\big(\rho^{\otimes n}\big\|\gamma^{A^n}\big)+D_{\max}\left(\omega^{R_n}\big\|\gamma^{R_n}\right)\;,
\ea
Now, it is well known (see e.g.~\cite{GT2021}) that all quantum relative entropies, in particular, $D_{\max}$, satisfy 
\begin{align*}
D_{\max}(\omega^{R_n}&\big\|\gamma^{R_n})\leq\log\left\|\left(\gamma^{R_n}\right)^{-1}\right\|_{\infty}\\
&=\log\left(\tr\left[e^{-\beta H^{R_n}}\right]\exp\left(\beta\left\|H^{R_n}\right\|_\infty\right)\right)\\
&\leq\beta \left\|H^{R_n}\right\|_\infty+\log\tr\left[e^{-\beta H^{R_n}}\right]\\
&\leq \beta cn^\alpha+\log|R_n|\;,
\end{align*}
where the last line follows from the fact that $(\omega^{R_n},\gamma^{R_n})$ is SLAR so there exists $c>0$ independent of $n$, and $\alpha\in[0,1)$ such that the maximal energy of system $R_n$ does not exceed $cn^\alpha$. Moreover, since $|R_n|$ is polynomial in $n$ we get that for sufficiently large $n$, $\log|R_n|\leq\beta cn^{\alpha}$. Taking the supremum over all possible SLAR systems $R_n$ we get that
for any $\eps>0$
\ba
&\varlimsup_{n\to\infty}\frac1n\distill_{\mf}^{\eps}\left(\rho^{\otimes n},\gamma^{A^n}\right)\\
&=\sup_{\{R_n\}_{n\in\mbb{N}}}\varlimsup_{n\to\infty}\frac1n\distill_{\gpc}^{\eps}\left(\rho^{\otimes n}\otimes\omega^{R_n},\gamma^{A^nR_n}\right)\\
&\leq \sup_{\alpha\in[0,1),c\in\mbb{R}_+}\varlimsup_{n\to\infty}\frac1n\left(D_{\min}^\eps\big(\rho^{\otimes n}\big\|\gamma^{A^n}\big)+2\beta cn^\alpha\right)\\
&=D(\rho^A\|\gamma^A)\;.
\ea
This completes the proof.
\end{proof}

\end{appendix}

\end{document}